\documentclass[review]{elsarticle}

\usepackage{hyperref}

\usepackage{amssymb,amsmath}
\journal{}
\usepackage{geometry}
\usepackage{graphicx}
\usepackage{booktabs}
\usepackage{color}
\usepackage{colortbl}

\newtheorem{definition}{Definition}
\newtheorem{theorem}{Theorem}
\newtheorem{proof}{Proof}

\usepackage{multirow}
\usepackage{graphicx}
\usepackage{subfigure}

\biboptions{authoryear}

\begin{document}

\begin{frontmatter}

\title{ Application of a new information priority accumulated grey model with time power to predict short-term wind turbine capacity \tnoteref{mytitlenote}}

\tnotetext[mytitlenote]{Please cite this article as: Jie Xia, Xin Ma, Wenqing Wu, Baolian Huang, Wanpeng Li.
Application of a new information priority accumulated grey model with time power to predict short-term wind turbine capacity. Journal of Cleaner Production, Volume 244, 2020, 118573, doi: https://doi.org/10.1016/j.jclepro.2019.118573.}%Fully documented templates are available in the elsarticle package on \href{http://www.ctan.org/tex-archive/macros/latex/contrib/elsarticle}{CTAN}.}

%%%%% use optional labels to link authors explicitly to addresses:
\author{Jie Xia \fnref{label1,label2}}
 \address[label1]{School of Mathematical Sciences, University of Electronic
                  Science and Technology of China, 611731, Chengdu, China}
 \address[label2] {School of Science, Southwest University of
                  Science and Technology, 621010, Mianyang, China}
 %==============================================================================================================
 \author[label2]{Xin Ma \corref{cor1}}
 \cortext[cor1]{Corresponding author. Xin Ma}
 \ead{cauchy7203@gmail.com}
 %%%%%%%%%%%%%%%%%%%%%%%%%%%%%%%%%%%%%%%%%%%%%%%%%%%%%%%%%%%%
\author{Wenqing Wu \fnref{label2}}
%------------------------------------------------------------------------
  \author{Baolian Huang \fnref{label3}}
  \address[label3]{School of Economics,
                   Wuhan University of Technology, 430070, Wuhan, China}
%%%%%%%%%%%%%%%%%%%%%%%%%%%%%%%%%%%%%%%%%%%%%%%%%%%%%%%%%%%%%%%%%%%%
\author{Wanpeng Li \fnref{label4}}
  \address[label4]{School of Computing, Mathematics and Digital Technology,
                   Manchester Metropolitan University, Manchester, United Kingdom}

%=================================================================
\begin{abstract}
Wind energy makes a significant contribution to global power generation.
Predicting wind turbine capacity is becoming increasingly crucial for cleaner production.
For this purpose, a new information priority accumulated grey model with time power is proposed to predict
short-term wind turbine capacity.
Firstly, the computational formulas for the time response sequence and
the prediction values are deduced by grey modeling technique and the definite integral trapezoidal approximation formula.
Secondly, an intelligent algorithm based on particle swarm optimization is applied to determine the optimal nonlinear parameters of the novel model.
Thirdly, three real numerical examples are given to examine the accuracy of the new model by comparing with six existing prediction models.
Finally, based on the wind turbine capacity from 2007 to 2017, the proposed model is established to predict the total wind turbine capacity in Europe, North America, Asia, and the world. The numerical results reveal that the novel model is superior to other
 forecasting models. It has a great advantage for small samples with new characteristic behaviors.
 Besides, reasonable suggestions are put forward from the standpoint of the practitioners and governments, which has high potential to advance the sustainable improvement of clean energy production in the future.
\end{abstract}
%==========================================================
%=============================================
\end{frontmatter}

%%%%%%%%%%%%%%%%%%%%%%%%%%%%%%%%%%%%%%%%%%%%%%%%%%%%%%%
% main text

%%-------------------------------------------------------------------
\begin{table}[!htbp]
 \label{table:Nomenclature}
 \renewcommand{\baselinestretch}{1.25}
 {
 \tabcolsep=6pt
 \begin{tabular}{llllllllllllllll}
 {\bf{Nomenclature}}\\
 \hline
GM(1,1)                         &the basic grey system model\\
NGM(1,$N)$                      & non-linear grey multivariable models\\
NGM(1,1,$k$)                    & non-homogeneous exponential grey model\\
GMCN(1,$N$)                     & new information priority accumulated grey multivariable convolution model \\
NGM(1,1,$k,c$)                  & extended non-homogeneous exponential grey model \\
GAGM(1,1)                       & non-equidistance generalized accumulated grey forecasting model\\
DGM(1,1)                        & discrete grey model\\
NIPDGM(1,1)                     & new information priority accumulated discrete grey model\\
GMC(1,$N$)                      & convolution integral grey prediction model\\
GRM(1,1)                        & non-equidistant grey model based on reciprocal accumulated generating\\
FAGM(1,1)                       & fractional-order grey model\\
FDGM                            & fractional multivariate discrete grey model\\
NIPGM$(1,1,{t}^{\alpha })$      & new information priority accumulated grey model with time power\\
1-AGO                           & first-order accumulative generation operation \\
1-IAGO                          & first-order inverse accumulative generation operation \\
1-NIPAGO                        & first-order new information priority accumulated generation operation\\
1-NIPIAGO                       &first-order new information priority inverse accumulated generation operation\\
${{S}^{0}}$                     &the non-negative original sequence\\
${{S}^{1}}$                     &first-order new information priority accumulated generation operation sequence\\
${{S}^{-1}}$                    &first-order new information priority inverse accumulated generation operation\\
 $H^{1}$                        & the mean generation sequence with consecutive neighbors \\
$\lambda$                       & the accumulation generation parameter\\
$a,b,c $                        & the parameters of the grey system\\
$s_{k}^{0}$                     & the observational data for the system input at time $k$ \\
$s_{k}^{1}$                     & the 1-NIPAGO data at time $k$ \\
$s_{k}^{-1}$                    & the 1-NIPIAGO data at time $k$ \\
$\hat{s}_{k}^{0}$               & the prediction value at time $k$ \\
$\hat{s}_{k}^{1}$               &the time response value at time $k$ \\
$h_{k}^{1}$                     & the background value at time $k$ \\
${\rm APE}$                     & the absolute percentage error\\
${\rm RMSEPR}$                  & the root mean square error of prior-sample\\
${\rm RMSE}$                    &the root mean square error\\
${\rm RMSEPO}$                  &the root mean square error of post-sample\\
${\rm PR}(n)$                   &polynomial regression model\\
PSO                             &particle swarm optimization\\
${\rm ARIMA}(p,d,q)$            &autoregressive integrated moving average model\\
CWEC                            & Global Wind Energy Council \\
\hline
\end{tabular}}
\end{table}

\newpage
\section{Introduction}
\label{sec1:intro}

With the global energy crisis and environmental pollution, clean and renewable energy has received extensive attention in recent years\citep{Pali2018A,LU201968}.
Wind energy is one of the most rapidly growing clean energies, which produces a great deal of electrical
energy by wind turbines  \citep{SHOAIB2019346}.
Wind power generation accounts for an increasing proportion in the global power production structure\citep{Jr2018Comparison}.
As reported by BP in {\it Statistical Review of World Energy} 2018, global wind turbine capacity growth averaged 20.2\% per year in the past decade.
And the wind turbine capacity of Asia, Europe, and North America accounted for 95.6\% of the world, while other
regions accounted for only 4.4\% in 2017.
Recently, the Global Wind Energy Council(CWEC) said that the wind power capacity of the world is expected to increase by 50\% and exceed more than 300 million kilowatts by 2023.
Wind energy has significant advantages and good development prospects in the development of cleaner production\citep{Kiaee2018Utilisation}.
The reason is that it can reduce carbon dioxide emissions and fossil fuels burning\citep{13285440920190110,MA2019915}.
Therefore, it is an inevitable choice for the global long-term energy strategy to develop wind energy, which can ensure sufficient energy supply\citep{Zeng2018shale}.
Hence, accurately predicting the wind turbine capacity in Asia, Europe, North America, and the world is very important for decision-makers.

In the previous studies, many scholars have proposed many models to forecast the wind turbine capacity, including logistic
model\citep{Shafiee2015Maintenance}, autoregressive sliding average model\citep{Wen2012Wind},
time series analysis\citep{Safari2018A}, support vector regression\citep{ZENDEHBOUDI2018272}, neural network prediction
model\citep{Chang2017An}, combined forecasting model\citep{Buhan2015Multi},
grey machine learning\citep{Ma2019machine,Wang2018nonlinear}, and grey model\citep{Wu2018Identifying,ZENG2019385}.
Among the many forecasting methods, the regression analysis method  uses the indicators related to the wind turbine capacity to build the model, which requires lots of samples. The calculation principle of time series model is
simple but can not reflect its intrinsic influencing factors.
The artificial neural network forecasting model has
excellent predictive ability for nonlinear data. But it is difficult to search the optimal solution so that it
cannot meet the accuracy requirements.
However, the grey model differs from other forecasting models that it requires a small sample with just 4 data or more.
Collecting sufficient samples is challenging in practical applications. Thus, more and more scholars have extensively concerned the grey system model.

Grey system theory along with grey models are initially put forward by \citet{Deng1982Control} to solve uncertain problems.{
 Because of the practicability of the grey model, the grey system has become a research direction with distinctive characteristics.
 The classical GM$(1,1)$ model has been generalized to other effective grey forecasting models,
including NGM$(1,1,k)$\citep{Cui2009Novel}, DGM$(1,1)$\citep{Rong2009Application},  NGM(1,$N)$ \citep{11978737120170121},
and CFGM\citep{MA2019cfgm}.
These grey models have been successfully applied in the environment\citep{Wu2018Using}, economy\citep{YIN2018815}, energy\citep{Wu2018energy}
 and other related fields\citep{Ding2018production,DUAN2019104853}.
 From the idea of GM $(1,1)$ modeling, it is the least-squares modeling method that follows the law of accumulated grey-index.
  And the traditional GM$(1,1)$ model has great forecasting effect on the data with homogeneous exponential law, and improved models
 have this characteristic.
There are a large number of systematic development laws that do not conform to the exponential law in real
 life. For the data with partial exponential features and time power terms, \citet{qian2012} constructed a novel GM $(1, 1,{t}^{\alpha })$ model.
 However, these grey models are built by using the first-order accumulative generation operation (1-AGO)\citep{Deng1982Control}.
And the restored values of these models are deduced by using the first-order inverse accumulative generation operation (1-IAGO).
Therefore, sequence accumulation generation is one of the critical steps of grey information mining and modeling.

In many references, the research on grey accumulation generation is mainly divided into two categories.

{\bf 1)} The idea of accumulation generation is combined with other forecasting models.
\citet{Sheng2008A} put forward a GSVMG model,
and it was used to forecast patent application filings, which obtained higher prediction accuracy.
\citet{Liu2011Ship} developed a GMRBF $(2, 1)$ model
by combining the RBF neural network with grey accumulation generation, which was successfully applied to forecast ship carrying capacity.
Recently, it is noteworthy that \citet{Zhou2017New} defined the new accumulated generation operator with a parameter. And then a NIPDGM(1,1) was constructed, which obtained an excellent prediction effect in energy prediction of Jiangsu province.
Later, \citet{Wu2018Grey} proposed the GMCN(1, $N)$ model by combining the new information priority principle with GM(1, $N$) model.

{\bf 2)} The expansion of the grey accumulation generation technology. \citet{Jie2010Non} established a new grey GAGM(1, $1)$ model by
 generalized accumulative generation operation.
This model was suitable for the unequal spacing sequences with the jumping trend and multistage.
 Based on the parallel number cumulative generation operation, the new grey GRM$(1,1)$ model was proposed by \citet{Xiao2012Grey}, which had practicality and reliability. To reduce the perturbation of the grey model solution, \citet{Wu2013Grey} introduced a fractional-order accumulation method and constructed a new FAGM(1,1) model.
 Later, a new seasonal discrete grey prediction model with periodic effects was proposed by \citet{XIA2014119}, which was successfully applied to forecast fashion consumer goods. Recently, \citet{Ma2019fractional} constructed the FDGM model and optimized it with the Gray Wolf algorithm.

Through the review and analysis of the above literature, it can be noticed that many models do not consider the new information priority principle. This may be the reason for the poor prediction accuracy.
 In order to solve this challenge for grey GM$(1,1,{t}^{\alpha })$ model, this study constructs a novel new information priority accumulated grey model with time power.
Furthermore, the computational formulas for the sequence of time response and the values of prediction are deduced.
 Another problem of the current grey model with new information priority accumulation\citep{Wu2018Grey} is that no detailed optimization algorithm has been used to seek the optimum solution of parameters.
 Therefore, we establish an optimization model to search the parameters and use the PSO algorithm to determine the optimized values of the novel model.
 Then, the novel model is applied to predict wind turbine capacity in Europe, North America, Asia, and the world. The numerical calculation results are compared with several existing models.
 Finally, according to the prediction results of wind turbine capacity in these regions from 2018 to 2020, reasonable suggestions on clean energy production are provided.

The remainder of this research is structured as below: Section \ref{sec:NIPGM11} systematically discusses the novel grey model.
 Section \ref{sec:model-check} gives how to optimize the nonlinear parameters of the novel model by an intelligent algorithm.
Section \ref{sec:validation} validates the accuracy of the novel model through three real cases. Section \ref{sec:applications} predicts wind turbine capacity by using seven forecasting models, and Section \ref{sec:conclu} gives the conclusions of the study.

%=======================================================================
%%%%%%%%%%%%%%%%%%%%%%%%%%%%%%%%%%%%%%%%%%%%%%%%%%%%%%%%%%%%%%%%%%%%%%%%
\section{The new information priority accumulated grey model with time power}
\label{sec:NIPGM11}

\subsection{Definition of new information priority accumulation}
 \label{subsec:priority-accumulation}

 \begin{definition}
\label{defiction1}
(see \citet{Zhou2017New})
 Set the non-negative historical data sequence as
${{S}^{0}}$ $=\left\{ s_{k}^{0}|k\in 1,2,\cdots ,m \right\}$,
 the first-order new information priority accumulated generation operation sequence (1-NIPAGO) of ${{S}^{0}}$ is
 ${{S}^{1}}=\left\{ s_{k}^{1}|k\in 1,2,\cdots ,m \right\}$. There is
\begin{equation}
s_{k}^{1}=\sum_{v=1}^{k} \lambda^{k-v} s_{v}^{0}, \quad \lambda \in(0,1), k\in 1,2,3, \cdots ,m,
 \label{Eq:num2}
\end{equation}
and $\lambda$ presents the accumulation generation parameter, which is used to adjust the weight of the sequence.
 Eq.(\ref{Eq:num2}) is named new information priority accumulation.
 \end{definition}

 In previous studies, \citet{Wu2018Grey} proved that the weight of ``new" 1-NIPAGO series $s_{k}^{1}$ is larger than the ``old" ones.
Being similar to the traditional grey model accumulation generation, ${{S}^{0}}$ can be accumulated and generated multiple times according to the accumulation mentioned above,
and the multiple new information of ${{S}^{0}}$ can obtain the priority accumulation generation sequence ${{S}^{n}}$.

\begin{definition}
\label{defiction2}
(see \citet{Zhou2017New})
Assuming the first-order new information priority inverse
accumulated generation operation sequence (1-NIPIAGO) of ${{S}^{0}}$ is  ${{S}^{-1}}=\left\{ s_{k}^{-1}|k\in 1,2,\cdots ,m \right\}$, where
\begin{equation}
\left\{\begin{array}{l}{s_{k}^{-1}=s_{k}^{0}-\lambda s_{k-1}^{0}} \\ {s_{1}^{-1}=s_{1}^{0}}\end{array}\right., \quad \lambda \in(0,1), k\in 2,3,4, \cdots, m.
 \label{Eq:num3}
\end{equation}
\end{definition}

It is worth noting that new information priority accumulated and new information priority inverse accumulated have the following relationship, these is
\begin{equation}
\begin{aligned} s_{k}^{1}-\lambda s_{k-1}^{1} &=\sum_{v=1}^{k} \lambda^{k-v} s_{v}^{0}-\lambda \sum_{v=1}^{k-1} \lambda^{k-v-1} s_{v}^{0}  &=s_{k}^{0}.
\end{aligned}
 \label{Eq:num4}
\end{equation}

%%%%%%%%%%%%%%%%%%%%%%%%%%%%%%%%%%%%
 The Eq.(\ref{Eq:num4}) is particularly important when establishing grey forecasting models and calculating prediction values for the original sequence,
as shown below.

\subsection{The definite integral trapezoidal approximation formula}
 \label{subsec:Trapezoidal-approximation}
 In this subsection, the problem of estimating the value of the integral $\int_{{{d}_{1}}}^{{{d}_{2}}}{f(u)du}$ is discussed.
 First, the interval $\left[ {{d}_{1}},{{d}_{2}} \right]$ is divided into $k$ sub-intervals of width
  $\Delta d=\frac{{{d}_{2}}-{{d}_{1}}}{k}$, and each sub-interval can be expressed as:
  $\left[ {{\gamma }_{0}},{{\gamma }_{1}} \right], \cdots, \left[ {{\gamma }_{k-1}},{{\gamma }_{k}} \right]$,
  where ${{\gamma }_{0}}={{d}_{1}}$ , ${{\gamma }_{k}}={{d}_{2}}$ .

Further, let the function value of each point ${{\gamma }_{0}},{{\gamma }_{1}},{{\gamma }_{2}},\cdots,{{\gamma }_{k}}  $ of $f(u)$ correspond to $f(\gamma_{0}), f(\gamma_{1}), f(\gamma_{2})$, $\cdots,$ $ f(\gamma_{k})$, as shown in Fig. \ref{fig:figure1}.
%%%----------------------------------------------------------------------
\begin{figure}[!htbp]
\centering
\includegraphics[height=6.1cm,width=10cm]{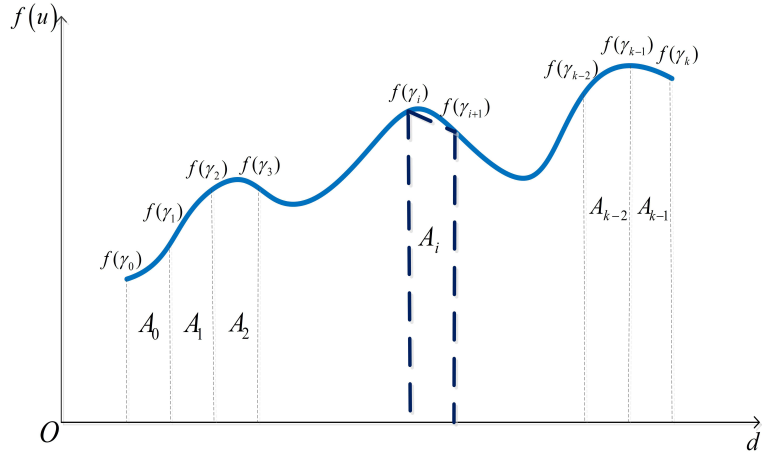}
\caption{The diagram of trapezoidal formula to estimate curve integral.}
 \label{fig:figure1}
\end{figure}

Then the area of each narrow trapezoid is:
\begin{equation}
A_{i}=\frac{\Delta d}{2}\left(f\left(\gamma_{i}\right)+f\left(\gamma_{i+1}\right)\right).
\label{Eq:num5}
\end{equation}

Therefore, if there are $k$ sub-intervals, the integral can be approximated as
\begin{equation}
\int_{d_{1}}^{d_{2}} f(u) d u \approx \frac{\Delta d}{2} \sum_{i=0}^{k-1}\left(f\left(\gamma_{i}\right)+f\left(\gamma_{i+1}\right)\right),
\label{Eq:num6}
\end{equation}
which is the definite integral trapezoidal approximate formula.

%%%%%%%%%%%%%%%%%%%%%%%%%%%%%%%%%%%%%%%%%%%%%%%%%%%%%%%%%%%%%%%%%%%%%%%%%%%%%%%%%%%%%%
\subsection{Modeling process of NIPGM$(1,1,{t}^{\alpha })$ model}
 \label{subsec:NIPGM11-new}
The tradition grey model with time power was established by \citet{qian2012}. In the following, the novel new information priority accumulated grey model with time power is defined as.

\begin{definition}
\label{defiction3}
Set the non-negative original  sequence as ${{S}^{0}}$,
the 1-NIPAGO of ${{S}^{0}}$ is  ${{S}^{1}}$, ${{s}_{k}^{1}}$ is shown in Eq.(\ref{Eq:num2}).
The mean generation sequence with consecutive neighbors is $H^{1}=\left\{h_{k}^{1} | k \in 2,3, \cdots, m\right\}$,
where $h_{k}^{1}=0.5 s_{k}^{1}+0.5 s_{k-1}^{1}$.
\end{definition}

\begin{definition}
\label{defiction4}
 Set ${{S}^{0}}$, ${{S}^{1}}$ and ${{H}^{1}}$ be shown in Definition \ref{defiction3}, there is
\begin{equation}
s_{k}^{0}+a h_{k}^{1}=b k^{\alpha}+c,
\label{Eq:num7}
\end{equation}
is named the mathematical form of NIPGM(1,1,${{t}^{\alpha }}$), then,
\begin{equation}
\frac{d s_{t}^{1}}{d t}+a s_{t}^{1}=b t^{\alpha}+c,
\label{Eq:num8}
\end{equation}
which is named the whitening equation of NIPGM (1,1,${{t}^{\alpha }}$). $\alpha $ presents a non-negative constant, $a$ presents development coefficient,  and the amount of grey action is $b{{t}^{\alpha }}+c$.
\end{definition}

The Eq.(\ref{Eq:num8}) is integral on interval $\left[ k-1,k \right]$, there is
\begin{equation}
s_{k}^{1}-s_{k-1}^{1}+0.5 a\left(s_{k}^{1}+s_{k-1}^{1}\right)=b \frac{k^{1+\alpha}-(k-1)^{1+\alpha}}{1+\alpha}+c,
\label{Eq:num10}
\end{equation}

Because of $h_{k}^{1}=0.5 s_{k}^{1}+0.5 s_{k-1}^{1} $, Eq.(\ref{Eq:num10}) turns to be
\begin{equation}
s_{k}^{1}-s_{k-1}^{1}+a h_{k}^{1}=b \frac{k^{1+\alpha}-(k-1)^{1+\alpha}}{1+\alpha}+c,
\label{Eq:num11}
\end{equation}

Further, there is
\begin{equation}
-a h_{k}^{1}+b \frac{k^{1+\alpha}-(k-1)^{1+\alpha}}{1+\alpha}+c=s_{k}^{1}-s_{k-1}^{1},
\label{Eq:num12}
\end{equation}

 \begin{theorem}
 \label{theorem1}
Suppose ${{S}^{0}}$, ${{S}^{1}}$ and ${{H}^{1}}$ are defined in Definition \ref{defiction3},
$\hat{r}={{\left[ a,b,c \right]}^{T}}$ is a parameter column, the least-squares parameter estimate of the novel model
 satisfies $\hat{r}={{({{F}^{T}}F)}^{-1}}{{F}^{T}}G$, where
\begin{equation}
F=\left( \begin{matrix}
   -h_{2}^{1} & \frac{{{2}^{1+\alpha }}-1}{1+\alpha } & 1  \\
   -h_{3}^{1} & \frac{{{3}^{1+\alpha }}-{{2}^{1+\alpha }}}{1+\alpha } & 1  \\
   \vdots  & \vdots  & \vdots   \\
   -h_{m}^{1} & \frac{{{m}^{1+\alpha }}-{{(m-1)}^{1+\alpha }}}{1+\alpha } & 1  \\
\end{matrix} \right),
G=\left( \begin{matrix}
   s_{2}^{1}-s_{1}^{1}  \\
   s_{3}^{1}-s_{2}^{1}  \\
   \vdots   \\
   s_{m}^{1}-s_{m-1}^{1}  \\
\end{matrix} \right).
\label{Eq:num13}
\end{equation}
 \end{theorem}

\begin{proof}
Using the method of the mathematical induction, take $k=2,3,\cdots,m$ into Eq.(\ref{Eq:num12}), there is
\begin{equation}
\left\{\begin{array}{l}{a\left(-h_{2}^{1}\right)+b \frac{2^{1+\alpha}-1}{1+\alpha}+c=s_{2}^{1}-s_{1}^{1}} \\
 {a\left(-h_{3}^{1}\right)+b \frac{3^{1+\alpha}-2^{1+\alpha}}{1+\alpha}+c=s_{3}^{1}-s_{2}^{1}} \\
 ~~~~~~~~~~~~~~~~~~~~~~~~~~~~~ \vdots   \\
 {a\left(-h_{m}^{1}\right)+b \frac{m^{1+\alpha}-(m-1)^{1+\alpha}}{1+\alpha}+c=s_{m}^{1}-s_{m-1}^{1}}\end{array}\right.
\label{Eq:num14}
\end{equation}

Convert the Eq.(\ref{Eq:num14}) into the matrix form, then
\begin{equation}
\left( \begin{matrix}
   -{h_{2}^{1}} & \frac{{{2}^{1+\alpha }}-1}{1+\alpha } & 1  \\
   -{h_{3}^{1}} & \frac{{{3}^{1+\alpha }}-{{2}^{1+\alpha }}}{1+\alpha } & 1  \\
   \vdots  & \vdots  & \vdots   \\
   -{h_{m}^{1}} & \frac{{{m}^{1+\alpha }}-{{(m-1)}^{1+\alpha }}}{1+\alpha } & 1  \\
\end{matrix} \right)\left( \begin{matrix}
   a  \\
   b  \\
   c  \\
\end{matrix} \right)=\left( \begin{matrix}
   s_{2}^{1}-s_{1}^{1}  \\
   s_{3}^{1}-s_{2}^{1}  \\
   \vdots   \\
   s_{m}^{1}-s_{m-1}^{1}  \\
\end{matrix} \right).
\label{Eq:num15}
\end{equation}

In summary
\begin{equation}
\hat{r}=\left( \begin{array}{l}{a} \\ {b} \\ {c}\end{array}\right)=\left(F^{T} F\right)^{-1} F^{T} G.
\label{Eq:num16}
\end{equation}
\end{proof}

\begin{theorem}
 \label{theorem2}
Suppose $F$, $G$, $\hat{r}$ are described in Theorem \ref{theorem1}, the sequence of time response $\hat{s}_{k}^{1}$ of the NIPGM(1,1,${{t}^{\alpha }}$) model is:
\begin{equation}
\hat{s}_{k}^{1}=\left(s_{1}^{0}-\frac{c}{a}\right) e^{-a(k-1)}+\frac{b}{2} e^{-a(k-1)} \sum_{\gamma=1}^{k-1} \left\{\gamma^{\alpha} e^{a(\gamma-1)}+(\gamma+1)^{\alpha} e^{a \gamma}\right\}+\frac{c}{a},\quad k=2,3, \cdots, m,
\label{Eq:num17}
\end{equation}
then the restored values of $\hat{s}_{k+1}^{0}$ can be deduced by using the 1-NIPIAGO,
\begin{equation}
\begin{aligned} \hat{s}_{k+1}^{0}=& e^{-a(k-1)}\left(e^{-a}-\lambda\right)\left\{s_{1}^{0}+\frac{b}{2} \sum_{\gamma=1}^{k-1}\left\{\gamma^{\alpha} e^{a(\gamma-1)}+(\gamma+1)^{\alpha} e^{\alpha \gamma}\right\}\right\} \\ &+\frac{b}{2} \left\{k^{\alpha} e^{-a} +(k+1)^{\alpha}\right\}+\frac{c}{a} \left( 1-e^{-a}\right),\quad k=2,3, \cdots, m. \end{aligned}
\label{Eq:num18}
\end{equation}
\end{theorem}

\begin{proof}
 We all know the general solution of linear non-homogeneous differential equation Eq.(\ref{Eq:num8}) is composed of the solution of its corresponding homogeneous equation plus one of its special solution. For Eq.(\ref{Eq:num8}),
its homogeneous form as:
\begin{equation}
\frac{d s_{t}^{1}}{d t}+ a s_{t}^{1}=0,
\label{Eq:num19}
\end{equation}
then we solve the general solution of Eq.(\ref{Eq:num19}) is  $\ln s_{t}^{1}=-a t+c$. Further,
simplification gives $s_{t}^{1}=Q e^{-a t}$. Taking $t=1$, there is $Q=e^{a t} s_{1}^{1}$, $c$ and $Q$ is a constant.

Using the constant variation method, we obtain the solution of Eq.(\ref{Eq:num8}) as
\begin{equation}
 s_{t}^{1}=Q(t) e^{-a t},
\label{Eq:num20}
\end{equation}

Because of $Q(t)= s_{t}^{1} e^{a t}$ and $s_{t}^{1}=s_{t}^{0}$, then $Q(1)=e^{a} s_{1}^{0}$.
Bringing Eq.(\ref{Eq:num20}) into Eq.(\ref{Eq:num8}), we get
\begin{equation}
\frac{d Q(t)}{d t}=e^{a t}\left(b t^{\alpha}+c\right),
\label{Eq:num22}
\end{equation}

Further, considering the integral of Eq.(\ref{Eq:num22}) on the interval $\left[ 1,t \right]$, we can obtain
\begin{equation}
\int_{1}^{t} d Q(u)=\int_{1}^{t} e^{a u}\left(b u^{\alpha}+c\right) d u,
\label{Eq:num23}
\end{equation}

From Eq.(\ref{Eq:num23}) know that
\begin{equation}
Q(t)-Q(1)=b \int_{1}^{t} u^{\alpha} e^{a u} d u+\frac{c}{a}\left(e^{a t}-e^{a}\right),
\label{Eq:num24}
\end{equation}

Thus, $Q(t)$ can be represented  as
\begin{equation}
\left\{\begin{array}{l}{Q(t)=Q(1)+b \int_{1}^{t} u^{\alpha} e^{a u} d u+\frac{c}{a}\left(e^{a t}-e^{a}\right)} \\
{Q(t)=e^{a t}s_{t}^{1}} \end{array}.\right.
\label{Eq:num25}
\end{equation}

According to Eq.(\ref{Eq:num25}) and the definite integral trapezoidal approximation formula Eq.(\ref{Eq:num6}),
the continuous-time response function $s_{t}^{1}$ can be derived from the following formula,
\begin{equation}
\begin{aligned}
s_{t}^{1} &=\left(Q(1)+b \int_{1}^{t} u^{\alpha} e^{a u} d u+\frac{c}{a}\left(e^{a t}-e^{a}\right)\right)e^{-a t} \\
%&=\left(e^{a} s_{1}^{0}+b \int_{1}^{t} u^{\alpha} e^{a u} d u+\frac{c}{a}\left(e^{a t}-e^{a}\right)\right)e^{-a t}\\
&=\left(s_{1}^{0}-\frac{c}{a}\right) e^{-a(t-1)}+b e^{-a(t-1)} \int_{1}^{t} u^{\alpha} e^{a(u-1)} d u +\frac{c}{a}\\
&=\left(s_{1}^{0}-\frac{c}{a}\right) e^{-a(t-1)}+\frac{c}{a} \\
&+\frac{b}{2} e^{-a(t-1)} \sum_{\gamma=1}^{t-1} \left\{\gamma^{\alpha} e^{a(\gamma-1)}+(\gamma+1)^{\alpha} e^{a \gamma}\right\}.
\end{aligned}
\label{Eq:num26}
\end{equation}

Finally, the continuous-time response function is challenging to calculate in actual applications\citep{MA2019600}, so we further discrete Eq.(\ref{Eq:num26}) to obtain the time response sequence $\hat{s}_{k}^{1}$ . And the prediction values of $\hat{s}_{k}^{0}$ can be solved by using Eq.(\ref{Eq:num4}).
\end{proof}

%==============================================================
%==============================================================

\subsection{Generality of the NIPGM(1, 1, ${{t}^{\alpha }}$) model}
 \label{subsec:special-cases}

 The novel grey model is a more extensive model, which combines new information priority accumulation with grey model with time power.
 Fig.\ref{fig:relationship-NIPGM11} displays that the relationship between the novel model and the classical GM(1, 1)model\citep{Deng1982Control}, non-homogeneous exponential NGM(1, 1, $k$) model \citep{Cui2009Novel},
extended non-homogeneous exponential NGM(1, 1, $k$, $c$) model\citep{Wang2014Extended}, and grey GM(1, 1, ${{t}^{\alpha }}$)model with time power\citep{qian2012}.
\begin{figure}[!htbp]
\centering
\includegraphics[height=6.1cm,width=10cm]{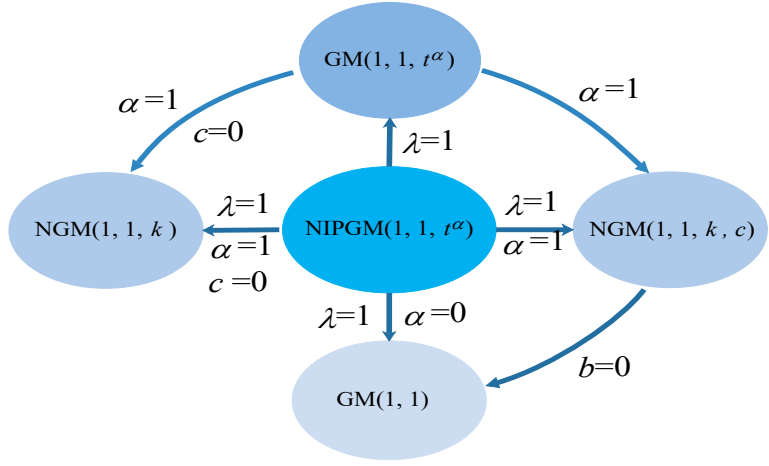}
\caption{The relationship of the novel model to other prediction models.}
 \label{fig:relationship-NIPGM11}
\end{figure}

{\bf \romannumeral1)} If $\alpha \text{=}0$ and $\lambda \text{=}1$, the proposed model becomes $s_{k}^{0}+a h_{k}^{1}=b k^{0}+c=b_{0}$,
which degenerates into the classical
GM(1, 1) model, then

1) The sequence of time response $\hat {s}_{k}^{1}$ as
\begin{equation}
\hat {s}_{k}^{1}=\left({s}_{1}^{1}-\frac{b_{0}}{a}\right) e^{-a (k-1)}+\frac{b_{0}}{a},
\label{Eq:num28}
\end{equation}

2)The values of prediction $\hat{s}_{k+1}^{0}$ as
\begin{equation}
\hat{s}_{k+1}^{0}=\left(1-e^{a}\right)\left({s}_{1}^{0}-\frac{b_{0}}{a}\right) e^{-a k}.
\label{Eq:num29}
\end{equation}

{\bf \romannumeral2)} If $\alpha \text{=1}$ and $\lambda \text{=}1$, the new model
 becomes $s_{k}^{0}+a h_{k}^{1}=b k+c$, which reduces into
 the NGM(1, 1, $k$, $c$) model, then
\begin{equation}
\int_{1}^{t} u e^{a(u-1)} d u=\frac{(a t-1)e^{a(t-1)}+(1-a)}{a^{2}},\nonumber
\end{equation}
there are

1) The sequence of time response $\hat {s}_{k}^{1}$ as
\begin{equation}
\hat {s}_{k}^{1}=\left({s}_{1}^{1}-\frac{a b + a c -b}{a^{2}}\right) e^{-a (k-1)}+\frac{b}{a} (k-1)+\frac{a c-b}{a^{2}},
\label{Eq:num31}
\end{equation}

2)The values of prediction $\hat{s}_{k+1}^{0}$ as
\begin{equation}
\hat {s}_{k+1}^{0}=\left(1-e^{a}\right)\left({s}_{1}^{0}-\frac{a b + a c - b}{a^{2}}\right) e^{-a k}+\frac{b}{a}.
\label{Eq:num32}
\end{equation}

{\bf \romannumeral3)} If $\alpha \text{=1}$,$\lambda \text{=}1 $ and $c=0$, the proposed model
becomes $s_{k}^{0}+a h_{k}^{1}=b k$, which degenerates into the NGM(1, 1, $k$) model,
there are

1) The sequence of time response $\hat {s}_{k}^{1}$ as
\begin{equation}
\hat {s}_{k}^{1}=\left({s}_{1}^{1}+\frac{b-a b}{a^{2}}\right) e^{-a (k-1)}+\frac{b}{a} (k-1)-\frac{b}{a^{2}},
\label{Eq:num34}
\end{equation}

2) The values of prediction $\hat{s}_{k+1}^{0}$ as
\begin{equation}
\hat {s}_{k+1}^{0}=\left(1-e^{a}\right)\left({s}_{1}^{0}+\frac{b-a b}{a^{2}}\right) e^{-a k}+\frac{b}{a}.
\label{Eq:num35}
\end{equation}

{\bf \romannumeral4)} If $\lambda \text{=}1$, the proposed model becomes
 $s_{k}^{0}+a h_{k}^{1}=b k^{2}+c$, which reduces into the GM(1, 1, ${{t}^{\alpha}}$) model with time power, then

1) The sequence of time response $\hat {s}_{k}^{1}$ as
\begin{equation}
\hat {s}_{k}^{1}=\left({s}_{1}^{0}-\frac{c}{a}\right) e^{-a(k-1)}+\frac{b}{2} e^{-a(k-1)} \sum_{\gamma=1}^{k-1} \left\{\gamma^{\alpha} e^{a(\gamma-1)}+(\gamma+1)^{\alpha} e^{a \gamma}\right\} +\frac{c}{a}, \quad  k\in 2,3, \cdots, m,
\label{Eq:num37}
\end{equation}

2) The values of prediction $\hat{s}_{k+1}^{0}$ as
\begin{equation}
\begin{aligned} \hat{s}_{k+1}^{0}=& e^{-a(k-1)}\left(e^{-a}-1\right)\left\{s_{1}^{0}+\frac{b}{2} \sum_{\gamma=1}^{k-1}\left\{\gamma^{\alpha} e^{a(\gamma-1)}+(\gamma+1)^{\alpha} e^{\alpha \gamma}\right\}\right\} \\ &+\frac{b}{2} \left\{k^{\alpha} e^{-a} +(k+1)^{\alpha}\right\}+\frac{c}{a} \left( 1-e^{-a}\right), \quad k\in 2,3, \cdots, m. \end{aligned}
\label{Eq:num38}
\end{equation}

%%%%%%%%%%%%%%%%%%%%%%%%%%%%%%%%%%%%%%%%%%%%%%%%%%%%%%%
It can be seen from Eq.(\ref{Eq:num29}), Eq.(\ref{Eq:num32}) and Eq.(\ref{Eq:num35}) that the classical GM(1, 1) model, NGM(1, 1, $k$, $c$) model, and NGM(1, 1, $k$) model
are suitable for sequences with non-homogeneous exponential law ${s}_{t}^{0} \approx d e^{at}+p$.
However, both Eq.(\ref{Eq:num18}) and Eq.(\ref{Eq:num38}) are composed of power functions and exponential functions.
The grey model with time power and the novel model can be applied to sequences with
partial exponential features and time power terms ${s}_{t}^{0} \approx d e^{at}+p t^{\alpha}+q$.
The difference between the grey model with time power and the new model is the way of sequence accumulation generation.
The GM(1, 1, ${{t}^{\alpha}}$) model does not successfully apply the new information priority principle.
Therefore, the proposed model is suitable for situations where the data characteristics are complex, which has more excellent  flexibility and practicality than the other four models.

%%%%=============================================================================
%%%===============================================================================
\section{Optimization of the parameters by particle swarm optimization}
\label{sec:model-check}

We can notice that the parameters $\lambda $ and $\alpha $ have been given before building the NIPGM (1,1,${{t}^{\alpha }}$) model.
Choosing the optimal parameters $\lambda $ and $\alpha $ is very important, which can enhance the fitting and predicting capabilities of the proposed model. The details about how to optimize the parameters $\lambda $ and $\alpha $ by the particle swarm optimization algorithm were given in this section.

\subsection{Model evaluation criteria}
 \label{subsec:model-checking}

To examine the accuracy of each forecasting model, we select the absolute percentage error (APE), the root mean square error of priori-sample (RMSEPR)\citep{Ma2017gmc}, the root mean square error of the post-sample (RMSEPO)\citep{Wu2018Application} and the root mean square error (RMSE)\citep{Ma2017gmc} as the assessment standard. These expressions are expressed as
\begin{eqnarray}
&& {\rm APE}(k)=\left|\frac{{s}_{k}^{0}-\hat {s}_{k}^{0}}{{s}_{k}^{0}}\right| \times 100 \%,
\label{Eq:num39}\\
&& {\rm RMSEPR}=\sqrt{\frac{1}{l} \sum_{k=1}^{l}\left({\rm APE}(k)\right)^{2}} \times 100 \%,
\label{Eq:num40}\\
&& {\rm RMSEPO}=\sqrt{\frac{1}{m-l} \sum_{k=l+1}^{m}\left({\rm APE}(k)\right)^{2}} \times 100 \%,
\label{Eq:num41}\\
&& {\rm RMSE}=\sqrt{\frac{1}{m} \sum_{k=1}^{m}\left({\rm APE}(k)\right)^{2}} \times 100 \%.
\label{Eq:num42}
\end{eqnarray}
where, $l$ presents the amount of sample applied to establish forecasting model, $m$ presents the total amount of sample.

%%%%%%%%%%%%%%%%%%%%%%%%%%%%%%%%%%%%%%%%%%%%%%%%%%%%%%%%%%%%%%%%%==============================================================
%%%%%%===================================================================================
\subsection{Nonlinear optimization model for the parameters $\lambda $ and $\alpha $}
 \label{subsec:optimal-parameters}
When using the NIPGM (1,1,${{t}^{\alpha }}$) model to forecast the original data, we first need to determine
the parameters $\lambda $ and $\alpha $ of the model, then use Eq.(\ref{Eq:num13}) to obtain the parameters $(a,b,c)$, and use Eq.(\ref{Eq:num18})
 to solve the prediction values $\hat {s}_{k}^{0}$.
In this paper, the minimum RMSE corresponding $\lambda $ and $\alpha $ are used as the optimal model parameters, and its objective function is as follows:
\begin{eqnarray}
&&\min _{\lambda, \alpha} {\rm RMSE}=\sqrt{\frac{1}{m} \sum_{k=1}^{m} \left(\frac{{s}_{k}^{0}-\hat {s}_{k}^{0}}{{s}_{k}^{0}}\right)} \times 100 \%,
\label{Eq:num43}\\
&& s.t.\left\{ {\begin{array}{*{20}l}
{0 \leq \lambda \leq 1, \alpha \geq 0,} \\
{(a, b, c)^{T}=\left(F^{T} F\right)^{-1} F^{T} G,} \\
{h_{k}^{1}=0.5 s_{k}^{1}+0.5 s_{k-1}^{1},}\\
{F=\left( \begin{matrix}
   -h_{2}^{1} & \frac{{{2}^{1+\alpha }}-1}{1+\alpha } & 1  \\
   -h_{3}^{1} & \frac{{{3}^{1+\alpha }}-{{2}^{1+\alpha }}}{1+\alpha } & 1  \\
   \vdots  & \vdots  & \vdots   \\
   -h_{m}^{1} & \frac{{{m}^{1+\alpha }}-{{(m-1)}^{1+\alpha }}}{1+\alpha } & 1  \\
\end{matrix} \right),
G=\left( \begin{matrix}
   s_{2}^{1}-s_{1}^{1}  \\
   s_{3}^{1}-s_{2}^{1}  \\
   \vdots   \\
   s_{m}^{1}-s_{m-1}^{1}  \\
\end{matrix} \right),}\\
{\hat{s}_{k}^{1}=\left(s_{1}^{0}-\frac{c}{a}\right) e^{-a(k-1)}+\frac{b}{2} e^{-a(k-1)} \sum_{\gamma=1}^{k-1} \left\{\gamma^{\alpha} e^{a(\gamma-1)}+(\gamma+1)^{\alpha} e^{a \gamma}\right\}+\frac{c}{a},}\\
{\hat{s}_{k}^{0}=\hat{s}_{k}^{1}-\lambda \hat{s}_{k-1}^{0},}\\
{k=2,3, \cdots, m}
\end{array}} \right. \nonumber
\end{eqnarray}

Since Eq.(\ref{Eq:num43}) is nonlinear, it is complicated to determine the values of parameters $\lambda$ and $\alpha$  by using Eq.(\ref{Eq:num43}).
For this reason, the optimal parameters $\lambda$ and $\alpha$ can be searched through the mature particle swarm optimization algorithm (PSO).
Inspired by the literature\citep{430837119891001,9841836520141001}, Fig.\ref{fig:figure3} gives the calculation flow chart which can clearly understand the modeling process.
\begin{figure}[!htbp]
\centering
\includegraphics[height=8.5cm,width=14cm]{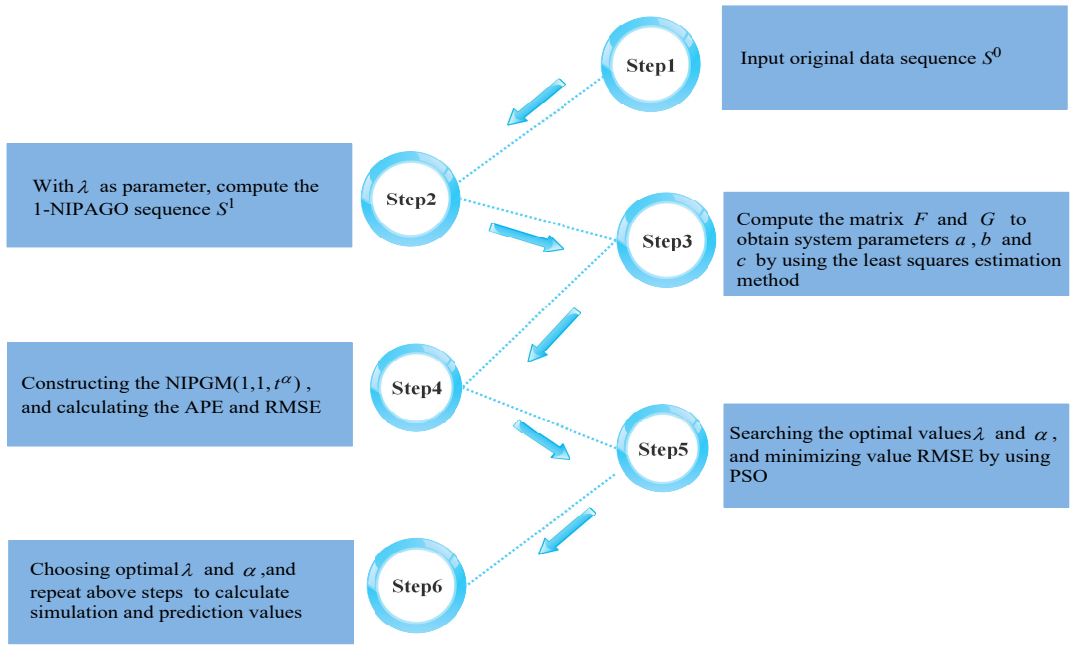}
\caption{Calculation steps for the NIPGM $(1,1,{{t}^{\alpha }})$ model.}
 \label{fig:figure3}
\end{figure}

\subsection{Optimization step of parameters}
 \label{subsec:PSO}

\citet{Shi2002Empirical} developed the particle swarm optimization (PSO) algorithm that originated  from the study of the behavior of bird predators.
This algorithm has many advantages, such as easy implementation, high precision, and rapid convergence. It has paid close attention to a large number of scholars and applied to many engineering fields \citep{Zeng2017optimal,Zeng2018background}.
 The specific algorithm steps will be given below.

 {\bf Step 1}: Set the parameters of learning factor ${{c}_{1}},{{c}_{2}}$, the weight of inertia ${{w}_{\min }}$, ${{w}_{\max }}$, and the maximum number of iterations ${{\rm iter}_{\max}}$.

 {\bf Step 2}: Initialize the particle swarm. Suppose that in an $D$-dimensional target space, there are $n$ candidate particles. The initial
 position $P_{i}$ and velocity $V_{i}$ of the $i$th particle are $P_{i}=\left(P_{i1}, P_{i2}, \ldots, P_{iD}\right),$
  $V_{i}=\left(V_{i1}, V_{i2}, \ldots, V_{iD}\right), i\in 1,2,3, \ldots, n$, respectively.

 {\bf Step 3:} Compute the values of fitness ${\rm RMSE}_{i}^{k}$} of each particle. Because the sum of $P_{i}$ and $V_{i}$ is not necessarily limited to the search space, some particles may exceed the edge of the search space during the particle swarm search process.
In order to settle this issue, we suppose that the penalty factor $M$ is an arbitrarily large constant and the penalty coefficient $C_{j}$
is used to determine whether the particle exceeds the limit value. If the optimal parameter exceeds the limit
value, then $C_{j}=1$, else $C_{j}=0$, there is
\begin{equation}
C_{j}=\left\{\begin{array}{ll}{1,} & {\text { parameter exceeds the limit value }} \\ {0,} & {\text { others }}\end{array}\right.
\label{Eq:costant_fitness}
\end{equation}

Therefore, the fitness of each particle is as follows:
\begin{equation}
{\rm RMSE}_{i}^{k}= {\rm RMSE}\left(P_{i j}^{k}\right)+\sum_{j=1}^{D} C_{j} M,  \quad  i\in 1,2,3, \ldots, n, \quad  j\in1,2,3, \ldots, D,
\label{Eq:fitness}
\end{equation}
where, $k$ presents current iteration.

{\bf Step 4}: Update ${\rm pbest}$ and ${\rm globebest}$. According to the fitness ${\rm RMSE}_{i}^{k}$ of each particle,
finding out the best position it experiences ${{b}_{k}}$, the individual extremum is recorded as $pbset_{i}={{b}_{k}}$.
The optimal position of the entire particle swarm is the global extremum, which recorded as ${\rm globebest}=P_{b_{k}}^{k}$.

(1)If ${\rm RMSE}_{i}^{k+1}<{\rm RMSE}_{i}^{k}$, then ${\rm pbest}_{i}^{k+1}=P _{i}^{k+1}$, else ${\rm pbest}_{i}^{k+1}={\rm pbest}_{i}^{k}$, and find
the optimal particle record as ${{b}_{k+1}}$.

(2)If ${\rm RMSE}_{{{b}_{k+1}}}^{k+1}<{\rm RMSE}_{{{b}_{k}}}^{k}$, then ${\rm globebest}^{k+1}={\rm pbest}_{{{b}_{k+1}}}^{k+1}$, else ${\rm globebest}^{k+1}={\rm globebest}^{k}$.

{\bf Step 5}: Update  $V_{ij}^{k+1}$ and  $P_{ij}^{k+1}$ according to the following mathematical expression:
\begin{eqnarray}
&& V_{ij}^{k+1}=V_{ij}^k\times w  + c_1 \times rand()\times
\left({\rm pbest}_{ij}^k-P_{ij}^k \right)
+ c_2 \times
rand()\times
\left({\rm globebest}_{j}^k -P_{ij}^k \right),
 \label{Eq:PSO-velocity} \\
&&P_{ij}^{k+1}=V_{ij}^{k+1} + P_{ij}^k,
 \label{Eq:PSO-position}\\
&& w=w_{\max}- \frac{w_{\max} -w_{\min}}{{\rm iter}_{\max}}\times k,
 \label{Eq:PSO-inertia-factor}
\end{eqnarray}
where, the inertia weight $w$ is in the interval $\left( {{w}_{\min }},{{w}_{\max }} \right)$, ${{c}_{1}}$ and ${{c}_{2}}$ are constant, and the symbol $rand()$ generates random number in $[0,1]$.

{\bf Step 6}: Stop the guidelines. If $k<{{\rm iter}_{\max}}$, continue to iterate back to step 3, else print optimum solution.

%%%%%%%%%%%%%%%%%%%%%%%%%%%%%%%%%%%%%%%%%%%%%%%%%%%%%%%%%%%%%%%
%%%%%%%%%%%%%%%%%%%%%%%%%%%%%%%%%%%%%%%%%%%%%%%%%%%%%%%%%%%%%%%

\section{Validation of the NIPGM$(1,1,{{t}^{\alpha }})$ model}
\label{sec:validation}
We present three real examples to examine the accuracy of the novel model, the prediction results are compared with four current grey forecasting models that are classical GM(1, 1)model, discrete DGM (1, 1)model, NGM(1, 1, $k$, $c)$ model, and GM(1, 1, ${{t}^{\alpha }})$ model.
Besides, we further assess the prediction accuracy of the novel model by comparing with other famous forecasting models including polynomial
regression model (PR)\citep{Zhou2018An}, autoregressive integrated moving average model (ARIMA)\citep{Jiang2018ARIMA}.
In the PR$(n)$ model, ``$n$" is the number of polynomial regressions.
In the ARIMA($p,d,q$)model, ``$p$ " presents the number of autoregressive terms, ``$d$ " presents the order of differences, and ``$q$ " presents the number of moving averages.

The unknown parameter $\lambda $ of grey model with time power and parameters $\lambda $, $\alpha $ of the proposed model are determined through the PSO algorithm. Table \ref{table:settings} gives the various parameters setting of the particle swarm algorithm. The computational results of the model optimization parameters $\lambda$, $\alpha$, and RMSE are stored in a matrix with a dimension of $100\times 1000$. Because each trail has an optimal value of $\lambda $, $\alpha $, and RMSE,
there are 100 optimal values of $\lambda $, $\alpha $, and RMSE. In this paper, the minimum RMSE and corresponding $\lambda $ and $\alpha $ of
the 100 parameters are selected as the system parameters of the model.
 \begin{table}[!htbp]
\caption{The parameters setting of PSO algorithm}
 \label{table:settings}
 \vspace{6pt}
 \renewcommand{\baselinestretch}{1.25}
 {\footnotesize\centerline{\tabcolsep=9pt
 \begin{tabular}{cccccccccccccccc}
 \hline
$c_{1}$  &&$c_{2}$    &&$w$   &&$n$  && ${{\rm iter}_{\max}}$   &&$M$ \\
 \hline
 2  && 2   &&0.6  &&100 &&1000  &&10000\\
\hline
\end{tabular} }}
\end{table}
%%%%%%%%%%%%%%%%%%%%%%%%%%%%%%%%%%%%%%%%%%%%%%%%%%%%%%%%%%%%%%%%%%%%%%%%%%
%%%%%%%%%%%%%%%%%%%%%%%%%%%%%%%%%%%%%%%%%%%%%%%%%%%%%%%%%%%%%%%%%%%%%%%%%%%%
\subsection{ Case1: Forecasting energy consumption of Jiangsu province}
\label{subsec:case1}
The testing data is derived from the literature\citep{Zhou2017New} in this case. We use scientific computing software Matlab to calculate
100 identical trails, and the numerical results are displayed in Table \ref{table:table1}. We see that the minimum
values for RMSE and the relevant $\alpha $ of the GM(1, 1, ${{t}^{\alpha }})$ model are 4.2431\%, 4.7427, the minimum values
for RMSE and the corresponding $\lambda $ and $\alpha $ of NIPGM$(1,1,{{t}^{\alpha }})$ model are 0.012\%, 0.3583, 0.6157,
respectively.
%------------------------------------------------------------------
\begin{table}[!htbp]
\caption{ The optimal parameters values of GM(1, 1, ${{t}^{\alpha }})$ model and the novel model in energy consumption of Jiangsu Province.}
 \label{table:table1}
 \vspace{6pt}
 \renewcommand{\baselinestretch}{1.25}
 {\footnotesize\centerline{\tabcolsep=9pt
 \begin{tabular}{ccccccc}
 \hline
     \multicolumn{3}{c}{NIPGM(1, 1, ${{t}^{\alpha }})$} && \multicolumn{2}{c}{GM(1, 1, ${{t}^{\alpha }})$} \\
    \hline
  RMSE(\%)      & $\lambda $      &$\alpha $      && RMSE(\%)    & $\alpha $      \\
 \hline
{{\bf0.0120}}   &0.3583    &0.6157 && 4.2431 & 4.7427   \\
\hline
\end{tabular} }}
\end{table}
%-------------------------------------------------------------------
%%%%%%%%%%%%%%%%%%%%%%%%%%%%%%%%%%%%%%%%%%%%%%%%%%%%%%%%%%%%%%%%%%%%%%%%%%%%%%%%%%%%%%%%%%%%%%

The prediction values of the NIPGM$(1,1,{{t}^{\alpha }})$ and GM$(1, 1,{{t}^{\alpha }})$ can be immediately obtained
when determining the optimal parameters. Then, we calculate all the numerical results of seven models
by using Matlab software, which are displayed in Table \ref{table:table2} and Table \ref{table:table3}.
%%%%%%%%%%%%%%%%%%%%%%%%%%%%%%%%%%%%%%%%%%%%%%%%%%%%%%%%%%%%%%%%%%%%%%%%%%%%%%%%%%%%%%%%%%%%%%%
\begin{table}[!htbp]
\caption{The fitted and predicted values for different models in energy consumption of Jiangsu Province(10000 tons of standard coal).}
 \label{table:table2}
 \vspace{6pt}
 \renewcommand{\baselinestretch}{1.25}
 {\footnotesize\centerline{\tabcolsep=3pt
 \begin{tabular}{ccccccccccccc}
 \hline
       Year      & value     &PR(2)   &ARIMA(2, 1, 0)  & GM(1, 1)	& DGM(1, 1)	 & NGM(1, 1, $k$, $c)$	 & GM(1, 1, ${{t}^{\alpha }})$	 & NIPGM(1, 1, ${{t}^{\alpha }})$\\
\hline
2001	& 8881   &8224.8818 	&8881.0000  & 8881.0000	&8881.0000	&8881.0000	&8881.0000	&8881.0000\\
2002	& 9593	 &10258.6212	&9593.0000  & 11349.6361	&11361.2997	&9622.0537	&10277.1738	&9590.5191\\
2003	& 11950	 &12275.3682	&11950.0000  &12644.2775	&12658.6033	&11984.7933	&11948.6092	&11949.8220\\
2004	& 14207  &14275.1227	&14206.9611   &14086.5972	&14104.0410	&14238.3553	&13860.9157	&14204.8033\\
2005	& 16360	 &16257.8848	&16359.2587   &15693.4409	&15714.5278	&16387.7846	&16008.5593	&16357.5988\\
2006	& 18412	 &18223.6545	&18413.0391   &17483.5756	&17508.9100	&18437.8931	&18356.5814	&18411.1219\\
2007	& 20369	 &20172.4318	&20368.7184   &19477.9090	&19508.1858	&20393.2700	&20829.7995	&20369.0361\\
2008	& 22235	 &22104.2167 	&22235.1366   &21699.7341	&21735.7513	&22258.2928	&23300.5917	&22235.4601\\
2009	& 24010	 &24019.0091 	&24014.3595   &24175.0005	&24217.6741	&24037.1366	&25574.9477	&24014.7434\\
2010	& 25711	 &25916.8091 	&25702.5611    &26932.6180	&26982.9982	&25733.7834	&27376.4503	&25711.3159\\[5pt]
2011	& 27329	 &27797.6167 	&27333.0685    &30004.7940	&30064.0841	&27352.0316	&28327.8106	&27329.5865\\
2012	& 28872	 &29661.4318 	&28871.8526    &33427.4100	&33496.9875	&28895.5037	&27929.5304	&28873.8775\\
\hline
\end{tabular} }}
\end{table}
%%%%%%%%%%%%%%%%%%%%%%%%%%%%%%%%%%%%%%%%%%%%%%%%%%%%%%%%%%%%%%%%%%%%%%%%%%%%%%%%%%%%%

 It is easy to see in Table \ref{table:table3} that
the RMSEPR of ARIMA(2,1,0) model is 0.0127. However, it is noteworthy that the RMSEPO, RMSE of NIPGM (1,1,${{t}^{\alpha }}$) model are
 0.0048,0.0120.
 The numerical results demonstrate that grey NIPGM (1,1,${{t}^{\alpha }}$) model
 has the highest forecasting accuracy in seven forecasting models.
 The reason is that the growth of the sequence shows a trend of slowing down and then accelerating, which conforms to the characteristics of the new information priority accumulation.
 Therefore, using the proposed model to simulate and predict this sequence has higher prediction accuracy.
\begin{table}[!htbp]
\caption{ Errors of different models in energy consumption of Jiangsu Province(\%).}
 \label{table:table3}
 \vspace{6pt}
 \renewcommand{\baselinestretch}{1.25}
 {\footnotesize\centerline{\tabcolsep=5pt
 \begin{tabular}{ccccccccccccc}
 \hline
       Year    &PR(2)    &ARIMA(2, 1, 0)  & GM(1, 1)	& DGM(1, 1)	 & NGM(1, 1, $k$, $c)$	 & GM(1, 1, ${{t}^{\alpha }})$	 & NIPGM(1, 1, ${{t}^{\alpha }})$\\
\hline
2001	&7.3879 	&0.0000  &0.0000	     &0.0000	&0.0000	&0.0000	&0.0000\\
2002	&6.9386 	&0.0000  &18.3116	&18.4332	&0.3029	&7.1320	&0.0259\\
2003	&2.7227 	&0.0000  &5.8099	&5.9297	&0.2912	&0.0116	&0.0015\\
2004	&0.4795 	&0.0003  &0.8475	&0.7247	&0.2207	&2.4360	&0.0155\\
2005	&0.6242 	&0.0045  &4.0743	&3.9454	&0.1698	&2.1482	&0.0147\\
2006	&1.0229 	&0.0056  &5.0425	&4.9049	&0.1406	&0.3010	&0.0048\\
2007	&0.9650 	&0.0014  &4.3747	&4.2261	&0.1192	&2.2623	&0.0002\\
2008	&0.5882 	&0.0006  &2.4073	&2.2453	&0.1048	&4.7924	&0.0021\\
2009	&0.0375 	&0.0182  &0.6872	&0.8649	&0.1130	&6.5179	&0.0198\\
2010	&0.8005 	&0.0328  &4.7513	&4.9473	&0.0886	&6.4776	&0.0012\\[5pt]						
2011	&1.7147 	&0.0149  &9.7910	 &10.0080	&0.0843	&3.6548	&0.0021\\
2012	&2.7342 	&0.0005  &15.7780    &16.0189  &0.0814	&3.2643	&0.0065\\[5pt]
RMSEPR	&2.5635 	&{\bf  \underline{0.0127}}   &7.1476	 &7.1742	&0.1885	&4.3974	 &0.0131\\
RMSEPO	&2.2821 	&0.0105   &13.1303 &13.3560	&0.0829	&3.4650	 &{\bf \underline{0.0048}}\\
RMSE	&2.5147 	&0.0124   &8.5525	 &8.6339	 &0.1741 &4.2431	&{\bf \underline{0.0120}}\\
\hline
\end{tabular} }}
\end{table}

%%%%%%%%%%%%%%%%%%%%%%%%%%%%%%%%%%%%%%%%%%%%%%%%%%%%%%%%%%%%%%%%%%%%%%%%%%%%%%%%%%%%%%%%%%%%%%%%%%%%%%%%%%%%%%%%%%%%

\subsection{ Case2: Forecasting the output values of the high technology industry}
\label{subsec:case2}
The sample data is taken from the literature\citep{Ding2017Multi} in this case. Firstly, establishing the grey prediction model by using the data
from 2005 to 2012 and testing the forecasting accuracy of grey models by using the data
from 2013 to 2014. Then, we calculate 100 identical trails by using Matlab software, and Table \ref{table:table4} gives the optimal numerical results.
As can be observed from Table \ref{table:table4}, the minimum values for RMSE and the relevant $\alpha $ of grey model with time power are
3.7990\%, and 1.7464. The minimum values for RMSE and the relevant $\lambda $ and $\alpha $ of NIPGM(1, 1, ${{t}^{\alpha }})$
model are 3.3509\%, 0.9741 and 1.6850, respectively.
%------------------------------------------------------------------
\begin{table}[!htbp]
\caption{ The optimal parameters values of GM(1, 1, ${{t}^{\alpha }})$ model and the novel model in the output values of the high technology industry.}
 \label{table:table4}
 \vspace{6pt}
 \renewcommand{\baselinestretch}{1.25}
 {\footnotesize\centerline{\tabcolsep=9pt
 \begin{tabular}{ccccccc}
  \hline
     \multicolumn{3}{c}{NIPGM(1, 1, ${{t}^{\alpha }})$} && \multicolumn{2}{c}{GM(1, 1, ${{t}^{\alpha }})$} \\
    \hline
  RMSE(\%)      & $\lambda $      &$\alpha $      && RMSE(\%)    & $\alpha $      \\
 \hline
{{\bf3.3509}}   &0.9741   &1.6850 && 3.7990 & 1.7464   \\
\hline
\end{tabular} }}
\end{table}

Once the optimal parameters are derived, the numerical results of the grey model with time power
and the novel model can be directly calculated.
Table \ref{table:table5} and  Table \ref{table:table6} list
the calculation results of seven models. We can observe from Table \ref{table:table6},
the RMSEPR of the extended non-homogeneous exponential grey
model is 2.8812, and the RMSEPO, RMSE of the new model are
2.4994, 3.3509. The numerical results reveal that the novel model has more excellent forecasting capability than
other prediction models. The reason is that the growing tendency of the sequence shows a trend of acceleration, then deceleration,
and acceleration again, has great volatility and conforms to the characteristics of the new information prioritized.
 Furthermore, only 10 data with non-stationary
 may be the reason that the PR($n$) model and ARIMA($p,d,q$) model have worse prediction results.
 So using the NIPGM$(1,1,{{t}^{\alpha }})$ model to stimulate and predict this sequence has higher prediction accuracy.
\begin{table}[!htbp]
\caption{The fitted and predicted values of different models in the output values of the high technology industry(One trillion yuan).}
 \label{table:table5}
 \vspace{6pt}
 \renewcommand{\baselinestretch}{1.25}
 {\footnotesize\centerline{\tabcolsep=4pt
 \begin{tabular}{ccccccccccccc}
 \hline
 Year  & value  &PR(2)    &ARIMA(1, 1, 1)  & GM(1, 1)	& DGM(1, 1)	 & NGM(1, 1, $k$, $c)$	 & GM(1, 1, ${{t}^{\alpha }})$	 & NIPGM(1, 1, ${{t}^{\alpha }})$\\
\hline
2005	&3.39	&3.5883     &3.3900    &3.3900	&3.3900	    &3.3900	    &3.3900  	&3.3900\\
2006	&4.16	&4.0483     &4.1346    &4.0566	&4.0650	    &4.2669	    &4.0401	    &4.1934\\
2007    &4.97	&4.6671     &4.9564    &4.7205	&4.7317	    &4.7987 	&4.6819	    &4.7482\\
2008	&5.57	&5.4448     &5.8125    &5.4932	&5.5077   	&5.4628  	&5.4379	    &5.5057\\
2009	&5.96	&6.3812     &6.2863    &6.3922	&6.4111	    &6.2920   	&6.3226	    &6.4382\\
2010	&7.45	&7.4764     &6.4872    &7.4385	&7.4626	    &7.3276	    &7.3535	    &7.5267\\
2011	&8.75	&8.7305     &8.6574    &8.6560	&8.6866	    &8.6209	    &8.5505	    &8.7567\\
2012	&10.23	&10.1433    &10.0763   &10.0727	&10.1113	&10.2358	&9.9369	    &10.1168\\[5pt]
2013	&11.60	&11.7150    &11.7224   &11.7214	&11.7697	&12.2525	&11.5393	&11.5974\\
2014	&12.74	&13.4455    &13.0781   &13.6399	&13.7000	&14.7710	&13.3884	&13.1903\\
\hline
\end{tabular} }}
\end{table}

%%%%%%%%%%%%%%%%%%%%%%%%%%%%%%%%%%%%%%%%%%%%%%%%%%%%%%%%%%%%%%%%%%%%%%%%%%%%%%%%%%%%%
%%%%%%%%%%%%%%%%%%%%%%%%%%%%%%%%%%%%%%%%%%%%%%%%%%%%%%%%%%%%%%%%%%%%%%%%%%%%%%%%%%%%%
\begin{table}[!htbp]
\caption{Errors of different models in the output values of the high technology industry(\%).}
 \label{table:table6}
 \vspace{6pt}
 \renewcommand{\baselinestretch}{1.25}
 {\footnotesize\centerline{\tabcolsep=5pt
 \begin{tabular}{ccccccccccccc}
 \hline
Year    &PR(2)    &ARIMA(1, 1, 1)  & GM(1, 1)	& DGM(1, 1)	 & NGM(1, 1, $k$, $c)$	 & GM(1, 1, ${{t}^{\alpha }})$	 & NIPGM(1, 1, ${{t}^{\alpha }})$\\
\hline
2005	&5.8505 	&0.0000  &0.0000	&0.0000	&0.0000	&0.0000	&0.0000\\
2006	&2.6843 	&0.6100  &2.4864	&2.2845	&2.5706	&2.8812	&0.8037\\
2007	&6.0937	    &0.2743  &5.0197	&4.7952	&3.4465	&5.7974	&4.4626\\
2008	&2.2484 	&4.3538  &1.3797	&1.1178	&1.9251	&2.3725	&1.1541\\
2009	&7.0670 	&5.4752  &7.2524	&7.5685	&5.5713	&6.0843	&8.0236\\
2010	&0.3547 	&12.9233 &0.1545	&0.1689	&1.6425	&1.2957	&1.0293\\
2011	&0.2231 	&1.0587  &1.0745	&0.7251	&1.4759	&2.2802	&0.0768\\
2012	&0.8472  	&1.5021  &1.5371	&1.1607	&0.0568	&2.8655	&1.1065\\[5pt]						
2013	&0.9914 	&1.0550  &1.0466	&1.4625	&5.6253	&0.5234	&0.0227\\
2014	&5.5375 	&2.6538  &7.0636	&7.5356	&15.9419 &5.0893	&3.5346\\[5pt]	
RMSEPR	&3.7840 	&5.6032   &3.5741	&3.5586	&{\bf \underline{2.8812}}	&3.7730	 &3.5569\\
RMSEPO	&3.9779 	&3.0643   &5.0492	&5.4279	&11.9538 &3.6176	&{\bf \underline{2.4994}}\\
RMSE	&3.8279 	&5.0324    &3.9498	&4.0493	&6.1815	&3.7390	  &{\bf \underline{3.3509}}\\
\hline
\end{tabular} }}
\end{table}

%%%%%%%%%%%%%%%%%%%%%%%%%%%%%%%%%%%%%%%%%%%%%%%%%%%%%%%%%%%%%%%%%%%%%%%%%%%%%%%%%%%%%%%%%%%%%%%%%
\newpage
\subsection{ Case3: Forecasting China's grain production}
\label{subsec:case3}
The sample data is collected from the literature in this case\citep{Ding2018Modeling}. We build forecasting models by employing the data from 2003 to 2012 and assess the prediction capability by employing the data from 2013 to 2015. Similar to
 case 1 and case 2, Table \ref{table:table7} displays the optimal numerical results.
\begin{table}[!htbp]
\caption{ The optimal parameters values of GM(1, 1, ${{t}^{\alpha }})$ model and the novel model in China's grain production.}
 \label{table:table7}
 \vspace{6pt}
 \renewcommand{\baselinestretch}{1.25}
 {\footnotesize\centerline{\tabcolsep=9pt
 \begin{tabular}{ccccccc}
 \hline
     \multicolumn{3}{c}{NIPGM(1, 1, ${{t}^{\alpha }})$} && \multicolumn{2}{c}{GM(1, 1, ${{t}^{\alpha }})$} \\
    \hline
  RMSE(\%)      & $\lambda $      &$\alpha $      && RMSE(\%)    & $\alpha $      \\
 \hline
{{\bf 0.9957}}   &0.8965   &0.0045 && 1.0867 & 1.1714    \\
\hline
\end{tabular} }}
\end{table}
%%%%%%%%%%%%%%%%%%%%%%%%%%%%%%%%%%%%%%%%%%%%%%%%%%%%%%%%%%%%%%%%%%%%%%

Furthermore, Table \ref{table:table8} and  Table \ref{table:table9} list all calculation results, we can observe the RMSEPR of the extended non-homogeneous exponential grey model is 0.8371, and the RMSEPO, RMSE of the proposed model are 0.6902, 0.9957. The numerical results reveal that the NIPGM (1, 1,${{t}^{\alpha }}$) model has more accurate stimulation and prediction
 accuracy. Because the growing tendency of the sequence has great volatility and the sequence has new characteristic behavior.
\begin{table}[!htbp]
\caption{The fitted and predicted values of different models in China's grain production(10000 tons). }
 \label{table:table8}
 \vspace{6pt}
 \renewcommand{\baselinestretch}{1.25}
 {\footnotesize\centerline{\tabcolsep=3pt
 \begin{tabular}{cccccccccc}
 \hline
Year  & value  &PR(2)    &ARIMA(2, 1, 1)  & GM(1, 1)	& DGM(1, 1)	 & NGM(1, 1, $k$, $c)$	 & GM(1, 1, ${{t}^{\alpha }})$	 & NIPGM(1, 1, ${{t}^{\alpha }})$\\
\hline
2003	&43069.50	&44309.6857 	&43069.5000  &43069.5000	&43069.5000	&43069.5000	&43069.5000	&43069.5000\\
2004	&46946.90	&45980.3040 	&46946.9000  &46813.7925	&46816.8336	&47263.8615	&46685.2880	&47382.7440\\
2005	&48402.20	&47624.1546 	&48066.8211  &48129.9249	&48133.1370	&48308.5300	&47970.6232	&47933.2230\\
2006	&49804.20	&49241.2375 	&49803.3114  &49483.0592	&49486.4496	&49450.0971	&49305.8670	&49167.2512\\
2007	&50160.28	&50831.5528 	&51210.9062  &50874.2359	&50877.8119	&50697.5506	&50686.0227	&50631.5469\\
2008	&52870.92	&52395.1004 	&51688.5287  &52304.5244	&52308.2939	&52060.7120	&52108.7434	&52179.9439\\
2009	&53082.08	&53931.8804 	&54124.0336  &53775.0243	&53778.9953	&53550.3139	&53572.8101	&53748.7258\\
2010	&54647.71	&55441.8927 	&54625.9259  &55286.8661	&55291.0470	&55178.0842	&55077.6047	&55305.8171\\
2011	&57120.80	&56925.1374 	&56033.2806  &56841.2122	&56845.6116	&56956.8388	&56622.8742	&56833.7626\\
2012	&58957.97	&58381.6145 	&58399.7137  &58439.2575	&58443.8843	&58900.5823	&58208.6070	&58322.6967\\[5pt]
2013 	&60193.84	&59811.3238 	&60310.0801  &60082.2305	&60087.0942	&61024.6184	&59834.9627	&59767.0199\\
2014 	&60702.60	&61214.2656 	&61615.3786  &61771.3944	&61776.5046	&63345.6700	&61502.2275	&61163.6716\\
2015 	&62143.90	&62590.4396 	&62208.5203  &63508.0478	&63513.4145	&65882.0115	&63210.7865	&62511.1658\\
\hline
\end{tabular} }}
\end{table}
%%%%%%%%%%%%%%%%%%%%%%%%%%%%%%%%%%%%%%%%%%%%%%%%%%%%%%%%%%%%%%%%%%%%%%%%%%%%%%%%%%%%%
%%%%%%%%%%%%%%%%%%%%%%%%%%%%%%%%%%%%%%%%%%%%%%%%%%%%%%%%%%%%%%%%%%%%%%%%%%%%%%%%%%%%%
\begin{table}[!htbp]
\caption{The errors of China's grain production by different models(\%). }
 \label{table:table9}
 \vspace{6pt}
 \renewcommand{\baselinestretch}{1.25}
 {\footnotesize\centerline{\tabcolsep=5pt
 \begin{tabular}{ccccccccccccc}
 \hline
 Year   &PR(2)   &ARIMA(2, 1, 1)  & GM(1, 1)	& DGM(1, 1)	 & NGM(1, 1, $k$, $c)$	 & GM(1, 1, ${{t}^{\alpha }})$	 & NIPGM(1, 1, ${{t}^{\alpha }})$\\
\hline
2003	&2.8795 	&0.0000  &0.0000	&0.0000	&0.0000	&0.0000	&0.0000\\
2004	&2.0589 	&0.0000  &0.2835	&0.2771	&0.6751	&0.5573	&0.9284\\
2005	&1.6075 	&0.6929  &0.5625	&0.5559	&0.1935	&0.8916	&0.9689\\
2006	&1.1304 	&0.0018  &0.6448	&0.6380	&0.7110	&1.0006	&1.2789\\
2007	&1.3383 	&2.0945  &1.4233	&1.4305	&1.0711	&1.0481	&0.9395\\
2008	&0.9000 	&2.2364  &1.0713	&1.0642	&1.5324	&1.4416	&1.3069\\
2009	&0.6009	    &1.9629  &1.3054	&1.3129	&0.8821	&0.9245	&1.2559\\
2010	&1.4533 	&0.0399  &1.1696	&1.1772	&0.9705	&0.7867	&1.2043\\
2011	&0.3425 	&1.9039  &0.4895	&0.4818	&0.2870	&0.8717	&0.5025\\
2012    &0.9776 	&0.9469  &0.8798   &0.8720 &0.0973  &1.2710  &1.0775\\[5pt]
2013	&0.6355 	&0.1931  &0.1854	&0.1773	&1.3802	&0.5962	&0.7091\\
2014	&0.8429 	&1.5037  &1.7607	&1.7691	&4.3541	&1.3173	&0.7596\\
2015	&0.7186 	&0.1040  &2.1951	&2.2038	&6.0153	&1.7168	&0.5910\\[5pt]
RMSEPR	&1.3519 	&1.4238  &0.9470	&0.9471	&{\bf \underline{0.8371}}	&1.0073	&1.0785\\
RMSEPO	&0.7373 	&0.8773  &1.6282	&1.6348	&4.3607	&1.2959	&{\bf \underline{0.6902}}\\
RMSE	&1.2275 	&1.3087  &1.1556	&1.1580	&2.2977	&1.0867	&{\bf \underline{0.9957}}\\
\hline
\end{tabular} }}
\end{table}
%%%%%%%%%%%%%%%%%%%%%%%%%%%%%%%%%%%%%%%%%%%%%%%%%%%%%%%%%%%%%%%%%%%%%%%%%%
%%%%%%%%%%%%%%%%%%%%%%%%%%%%%%%%%%%%%%%%%%%%%%%%%%%%%%%%%%%%%%%%%%%%%%%%%%%%%%%%%%%%%%%%%%%%%%

\subsection{Summary of the case studies}
Based on the results of the three case studies, we summarize the performance of simulation and prediction through seven forecasting models.
Table \ref{table:rank} gives the average values and ranks of RMSEPR, RMSEPO, and RMSE for seven models in all the cases.
It is easy to see that the improved grey model outperforms other models and has better generalization capabilities.
\begin{table}[!htbp]
\caption{The average values and ranks of RMSEPR, RMSEPO and RMSE for seven models in all the cases. }
 \label{table:rank}
 \vspace{6pt}
 \renewcommand{\baselinestretch}{1.25}
 {\footnotesize\centerline{\tabcolsep=4pt
 \begin{tabular}{lcccccccccccc}
 \hline
               &PR($n$)    &ARIMA($p$,$d$,$q$)   & GM(1, 1)	& DGM(1, 1)	 & NGM(1, 1, $k$, $c)$	 & GM(1, 1, ${{t}^{\alpha }})$	 & NIPGM(1, 1, ${{t}^{\alpha }})$\\
\hline
Average RMSEPR	&2.5665 	&2.3466 	&3.8896 	&3.8933 	&{\bf \underline{1.3023}} 	  &3.0592 	   &1.5495 \\
Simulation rank	        &4  &3    & 6	        &7	        & {\bf \underline{1}}	   & 5	       &  2\\
Average RMSEPO 	&2.3324 	&1.3174 	&6.6026 	&6.8062 	&5.4658 	      &0.7929 	    &{\bf \underline{1.0648}}\\
Prediction rank	  &3	   &2	       &6	        &7	          &5	            &4	        & {\bf \underline{1}}\\
Average RMSE 	&2.5234 	&2.1178 	&4.5527 	&4.6137 	&2.8844 	       &3.0229 	    &{\bf \underline{1.4529}}\\
Overall rank	&3	&2	    &6	        & 7	        & 4	           &5	        & {\bf \underline{1}}\\
\hline
\end{tabular} }}
\end{table}

On the one hand, the proposed model is compared with PR($n$) model and ARIMA($p,d,q$) model. It reveals that NIPGM(1, 1, ${{t}^{\alpha }})$ model has the most excellent capabilities of simulation and forecasting and has significant advantages for small samples.
The reason is that the ARIMA model requires extensive samples of more than 50 observations for prediction \citep{Shumway2017ARIMA}.
On the other hand, the RMSE of the NIPGM$(1, 1,{{t}^{\alpha }})$ model has been dramatically improved than the basic grey model
and discrete grey model. The result indicates that the introduction of the parameters $\alpha $ and $\lambda $ into the grey model is a
scientific and practical approach to heighten the forecasting capability of grey forecasting models.
 Then, we compare the accuracy of the NGM(1, 1, $k$, $c$) model, GM(1, 1, ${{t}^{\alpha }})$ model
 and novel grey model, which shows that new grey model has the highest prediction accuracy and the second-highest simulation accuracy. This illustrates the importance of new information priority accumulation for forecasting, and grey model with time power and non-homogeneous exponential grey model are particular cases of the proposed model.
Besides, the PSO algorithm is very stable when determining the parameters $\lambda $ and $\alpha $.
This is why we select the PSO to seek the optimal parameters of the new grey model.

%%%%%%%%%%%%%%%%%%%%%%%%%%%%%%%%%%%%%%%%%%%%%%%%%%%%%%%%%%%%%%%
%====================================================================
\section{ Applications in the wind turbine capacity}
 \label{sec:applications}
 This section discusses the wind turbine capacity of the world and the top three regions, which are Asia, Europe, and North America, respectively.
First of all, Table \ref{table:totaldata} lists the data of the wind turbine capacity from 2007 to 2017, which gathers from the {\it Statistical Review of World Energy 2018}. Secondly, the data is split into two parts,
the data from 2007 to 2014 is used to build the PR($n$) model, time series  model, and five grey models, including GM(1, 1) model, DGM (1, 1) model, NGM(1, 1, $k)$ model, GM(1, 1, ${{t}^{\alpha }})$ model, and NIPGM(1, 1, ${{t}^{\alpha }})$ model.
 Thirdly, the data from 2015 to 2017 is applied to assess the capability  of the fitting and forecasting. Then, we use Matlab software to compute results. Finally, Fig. \ref{fig:figure10} displays the structure chart of
forecasting the wind turbine capacity, which is drawn under the inspiration of
literature\citep{Ma2019economic,Du2019hybrid}.
\begin{table}[!htbp]
\caption{ The total wind turbine capacity of Europe, North America, Asia, and the world(Megawatts). }
 \label{table:totaldata}
 \vspace{6pt}
 \renewcommand{\baselinestretch}{1.25}
 {\footnotesize\centerline{\tabcolsep=3pt
 \begin{tabular}{ccccccccccccc}
 \hline
   Year   & Europe	&North America	 &Asia	 & World  \\
\hline
2007 	&56748.8850 	&18810.0000 	&15327.3260 	&91894.0080\\
2008 	&64943.4830 	&27940.0000 	&22356.3570     &116511.6230\\
2009 	&77019.9934 	&38933.0000 	&33737.5070 	&151655.8934\\
2010 	&86721.9742 	&45054.0000 	&48622.3270 	&182901.3012\\
2011 	&96603.1278 	&53485.0000 	&69073.8140 	&222516.8618\\
2012 	&109884.8729 	&67934.0000 	&87572.6850 	&269853.3279\\
2013 	&120994.6758 	&71093.0000 	&105496.3320 	&303112.5198\\
2014 	&133915.4447 	&78340.0000 	&129273.7820 	&351617.6747\\[3pt]
2015 	&147637.6457 	&87058.4200 	&167528.3270 	&417144.1127\\
2016 	&161939.8681 	&96994.0000 	&189684.6370 	&467698.4951\\
2017 	&178314.1463 	&104070.0000 	&209977.2340 	&514798.1313\\
\hline
\end{tabular} }}
\end{table}
%-------------------------------------------------------------------
\begin{figure}[!htbp]
\centering
\includegraphics[height=10cm,width=17cm]{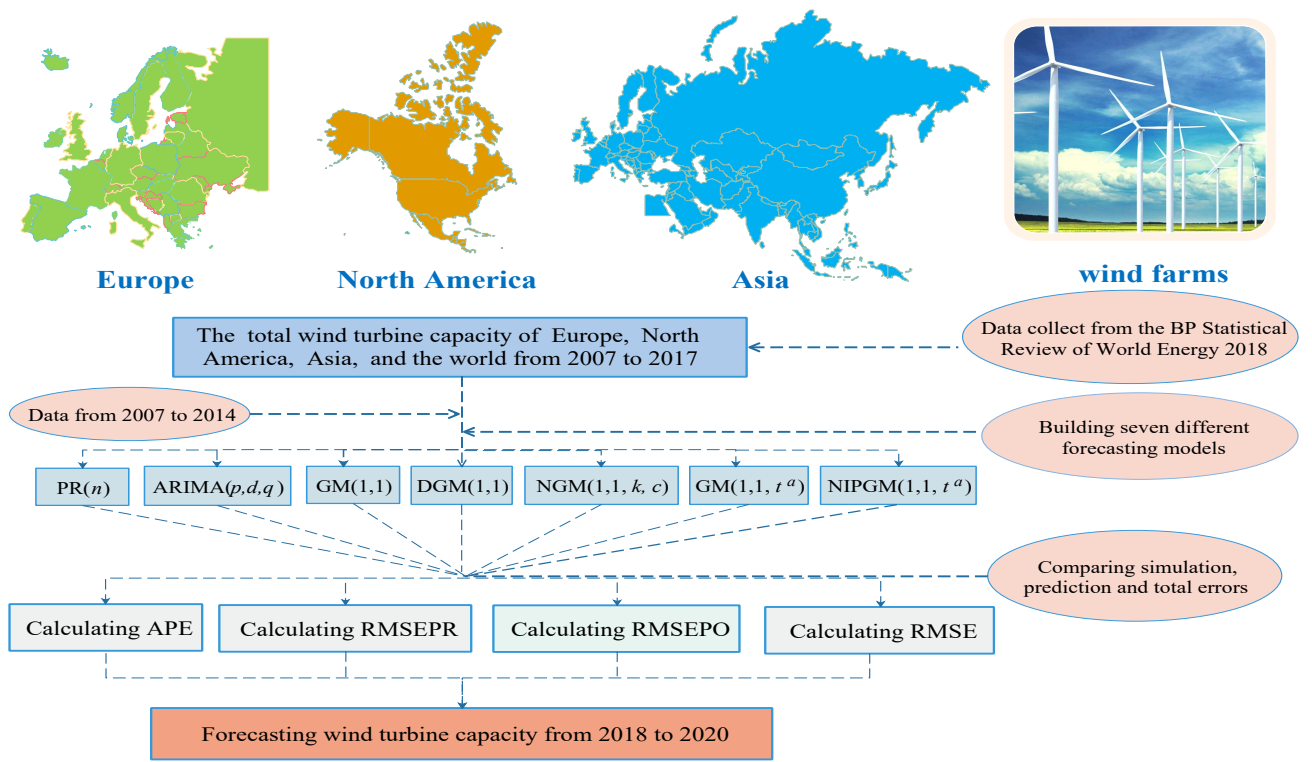}
\caption{The structure chart of forecasting wind turbine capacity.}
 \label{fig:figure10}
\end{figure}
%%%%%%%%%%%%%%%%%%%%%%%%%%%%%%%%%%%%%%%%%%%%%%%%%%%%%%%%%%%%%%%%%%%%%%
\subsection{Total wind turbine capacity of Europe}
 \label{subsec:Europe}
This subsection analyzes the total wind turbine capacity of Europe through seven prediction models.
First of all, through the demonstration in section 4, we use PSO algorithm to find the minimum RMSE and the corresponding $\lambda $, $\alpha $
of NIPGM $(1,1,{{t}^{\alpha }})$ model, the smallest RMSE and
the corresponding $\alpha $ of grey model with time power. Then, Table \ref{table:table10} displays the minimum RMSE and the relevant optimal values of the two models.
%------------------------------------------------------------------
\begin{table}[!htbp]
\caption{ The optimal parameters values of GM(1, 1, ${{t}^{\alpha }})$ model and the novel model in wind turbine capacity of Europe.}
 \label{table:table10}
 \vspace{6pt}
 \renewcommand{\baselinestretch}{1.25}
 {\footnotesize\centerline{\tabcolsep=9pt
 \begin{tabular}{ccccccc}
  \hline
     \multicolumn{3}{c}{NIPGM(1, 1, ${{t}^{\alpha }})$} && \multicolumn{2}{c}{GM(1, 1, ${{t}^{\alpha }})$} \\
    \hline
  RMSE(\%)      & $\lambda $      &$\alpha $      && RMSE(\%)    & $\alpha $      \\
 \hline
{{\bf 0.3799}}   &0.9649    &0.0206   && 1.6360 & 3.6598    \\
\hline
\end{tabular} }}
\end{table}
%%%%%%%%%%%%%%%%%%%%%%%%%%%%%%%%%%%%%%%%%%%%%%%%%%%%%%%%%%%%%%%%%%%%%%%%%%%%%%%%%%%%%

Further, five grey models are respectively constructed by using grey theory and the row data of the wind turbine capacity from 2007 to 2014.

$\bullet$ GM(1, 1) model

We can immediately get $(a,b)=(-0.1148,57660.2383)$ of the basic grey model by using the least-squares method. Thus the whitening equation is established, there is
\begin{equation}
\frac{d s_{t}^{1}}{d t}-0.1148 s_{t}^{1}=57660.2383.
\label{Eq:num47}
\end{equation}

$\bullet$ DGM(1, 1) model

We can get $({{\beta }_{1}},{{\beta }_{2}})=(1.1218,61171.4813)$ of DGM (1, 1)model.
Then the time response sequence is obtained, these is
\begin{equation}
\hat s_{k}^{1}=1.1218^{k} s_{k}^{0}-\frac{1-1.1218^{k}}{0.1218} \times 61171.4813, k=1,2, \cdots, m.
\label{Eq:num48}
\end{equation}

$\bullet$ NGM(1, 1, $k$, $c)$ model

We directly derive $(a,b,c)=(-0.0384,7583.1441$,$50768.5848)$ of the extended non-homogeneous exponential grey model.
And the whitening equation is established, then
\begin{equation}
\frac{d s_{t}^{1}}{d t}-0.0384 s_{t}^{1}=7583.1441 t+50768.5848.
\label{Eq:num49}
\end{equation}

$\bullet$ GM$(1, 1,{{t}^{\alpha }})$ model

We can observe from Table \ref{table:table10} that the optimal parameter $\alpha =3.6598$ of grey model with time power.
And applying the least-squares method, we can get $(a,b,c)=(-0.1396,-7.4040,53659.7544)$. Thus, the whitening equation
 is put forward, there is
\begin{equation}
\frac{d s_{t}^{1}}{d t}-0.1396 s_{t}^{1}=-7.4040 t^{3.6598}+53659.7544.
\label{Eq:num50}
\end{equation}

$\bullet$ NIPGM$(1,1,{{t}^{\alpha }})$ model

Similar to grey GM$(1, 1,{{t}^{\alpha }})$ model, $\lambda =0.9649$ and $\alpha =0.0206$ of NIPGM$(1,1,{{t}^{\alpha }})$ model are also obtained. Thus, the one-time new information priority accumulated generation operation sequence(1-NIPAGO) can be expressed as follow.
\begin{equation}
s_{k}^{1}=\sum_{v=1}^{k} 0.9649^{k-v} s_{v}^{0},k=1,2, \cdots, m.
\nonumber
\end{equation}

Then, the new information priority accumulation sequence $S^{1}$ and the matrix $F, G$ are given as follows.
\begin{equation}
{S^{1}}=\left( \begin{array}{*{35}{r}}
   56748.8850  \\
   119703.1115  \\
   192527.0720  \\
   272500.2666  \\
   359551.2611  \\
   456832.5442  \\
   561813.5647  \\
   676035.3844  \\
\end{array} \right),
F=\left( \begin{array}{*{35}{r}}
   -88225.9983 & 1.0080 & 1  \\
   -156115.0918 & 1.0190 & 1  \\
   -232513.6693 & 1.0261 & 1  \\
   -316025.7638 & 1.0315 & 1  \\
   -408191.9027 & 1.0358 & 1  \\
   -509323.0545 & 1.0394 & 1  \\
   -618924.4745 & 1.0425 & 1  \\
\end{array} \right),
G=\left( \begin{array}{*{35}{r}}
   62954.2265  \\
   72823.9605  \\
   79973.1945  \\
   87050.9945  \\
   97281.2831  \\
   104981.0205  \\
   114221.8197  \\
\end{array} \right).
\nonumber
\end{equation}

Finally, the optimal parameters $(a,b,c)=(-0.0737,-345863.1636,-291896.7690)$ of NIPGM$(1,1,{{t}^{\alpha }})$ model
are obtained. And the whitening equation is given, these is
\begin{equation}
\frac{d s_{t}^{1}}{d t}-0.0737 s_{t}^{1}=-345863.1636 t^{0.0206}-291896.7690.
\label{Eq:num51}
\end{equation}

%%%%%%%%%%%%%%%%%%%%%%%%%%%%%%%%%%%%%%%%%%%%%%%%%%%%%%%%%%%%%%%%%%%%%%%%

The prediction values of the seven models in the wind turbine capacity of Europe are listed in Table \ref{table:table11}.
 It is seen from Table \ref{table:table11} and Fig. \ref{fig:preEurope} that the established forecasting model can reflect the trend of wind turbine capacity in Europe from 2015 to 2017.
\begin{table}[!htbp]
\caption{The fitted and predicted values of different models in wind turbine capacity of Europe(Megawatts). }
 \label{table:table11}
 \vspace{6pt}
 \renewcommand{\baselinestretch}{1.25}
 {\footnotesize\centerline{\tabcolsep=3pt
 \begin{tabular}{ccccccccccccc}
 \hline
 Year  & data  &PR(3)    &ARIMA(2, 1, 1)  & GM(1, 1)	& DGM(1, 1)	 & NGM(1, 1, $k$, $c)$	 & GM(1, 1, ${{t}^{\alpha }})$	 & NIPGM(1, 1, ${{t}^{\alpha }})$\\
\hline	
2007	&56748.8850	&56446.3721 	&56748.8850 &56748.8850	&56748.8850	&56748.8850	&56748.8850	&56748.8850\\
2008	&64943.4830	&65969.7452 	&64943.4830 &68007.2739	&68086.5467	&65502.2396	&66034.0365	&64985.6025\\
2009	&77019.9934	&75999.9366	    &75106.1003 &76283.1268	&76381.7043	&75799.3255	&75715.7796	&76166.9682\\
2010	&86721.9742	 &86540.1637 	&86607.7451  &85566.0740	&85687.4822	&86499.7040	&86488.2844	&86832.2696\\
2011	&96603.1278	 &97593.6435 	&98309.8152  &95978.6696	&96127.0068	&97619.1706	&98253.2028	&97720.3058\\
2012	&109884.8729 &109163.5934   &108836.2802 &107658.3813	&107838.4052 &109174.1390	&110841.1855 &109111.8638\\
2013	&120994.6758 &121253.2304   &121799.3639 &120759.4053	&120976.6332 &121181.6661	&124006.1529 &121167.8665\\
2014	&133915.4447 &133865.7719   &134856.0580 &135454.7021	&135715.5250 &133659.4765 &137417.3283  &134005.8885\\[5pt]
2015	&147637.6457 &147004.4350   &147994.4670 &151938.2798	&152250.0937 &146625.9893 &150649.2690	&147726.0120\\
2016	&161939.8681 &160672.4369	&162535.6398 	&170427.7558 &170799.1109 &160100.3448	&163169.8761 &162421.7532\\
2017 &178314.1463 &174872.9948      &177791.4872 &191167.2290 &191608.0023  &174102.4332	&174326.2387 &178185.4696\\ %[5pt]
%2018	& &189609.3260	&194703.0777 	 &214430.5033	&214952.0940	&188652.9233	&183328.0850	&195111.4074\\
%2019	&	 &204884.6476	&212553.6713 &240524.7017	&241140.2560	&203773.2938	&189228.5472	&213297.6361\\
%2020	&	  &220702.1768 	&231314.2042 &269794.3215	&270518.9885	&219485.8644	&190901.8817	&232847.4225\\
\hline
\end{tabular} }}
\end{table}
%%%%%%%%%%%%%%%%%%%%%%%%%%%%%%%%%%%%%%%%%%%%%%%%%%%%%%%%%%%%%%%%%%%%%%%%%%%%%%%%%%%%
\begin{figure}[!htbp]
\centering
\includegraphics[height=8cm,width=12cm]{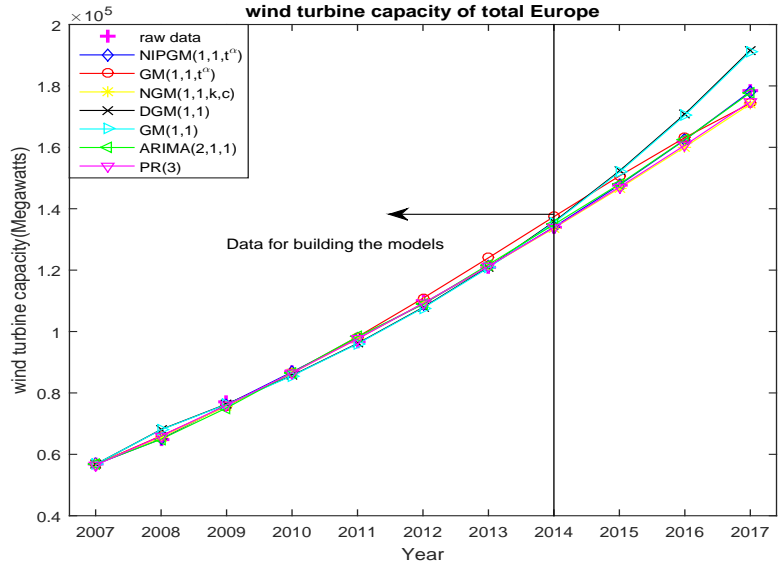}
\caption{The fitted and predicted values of different models in wind turbine capacity of Europe.}
 \label{fig:preEurope}
\end{figure}
%%%%%%%%%%%%%%%%%%%%%%%%%%%%%%%%%%%%%%%%%%%%%%%%%%%%%%%%%%%%%%%%%%%%%%%%%%%%%%%%%%%%%

The values of APE, RMSEPR, RMSEPO, and RMSE are calculated according to Eq.(\ref{Eq:num39}) to Eq.(\ref{Eq:num41}), and Table \ref{table:table12} gives all the results.
Besides, Fig. \ref{fig:figure11}(left) presents
the absolute percentage errors and Fig. \ref{fig:figure11}(right) provides the comparison results of RMSEPO and RMSE. It is easy observed from Table \ref{table:table12} and Fig.\ref{fig:figure11} that the RMSEPR, RMSEPO, and RMSE of NIPGM$(1,1,{{t}^{\alpha }})$ are 0.4815, 0.1432, 0.3799, respectively. The results display that
the novel model has the most beneficial fitting and forecasting capabilities than other models,
which explains the importance of new information priority accumulated to predicting.
%%%%%%%%%%%%%%%%%%%%%%%%%%%%%%%%%%%%%%%%%%%%%%%%%%%%%%%%%%%%%%%%%%%%%%%%%%%%%%%%%%%%%
\begin{table}[!htbp]
\caption{Errors of different models in  wind turbine capacity of Europe(\%). }
 \label{table:table12}
 \vspace{6pt}
 \renewcommand{\baselinestretch}{1.25}
 {\footnotesize\centerline{\tabcolsep=3pt
 \begin{tabular}{ccccccccccccc}
 \hline
 Year   &PR(3)    &ARIMA(2, 1, 1)  & GM(1, 1)	& DGM(1, 1)	 & NGM(1, 1, $k$, $c)$	 & GM(1, 1, ${{t}^{\alpha }})$	 & NIPGM(1, 1, ${{t}^{\alpha }})$\\
\hline
2007	&0.5331 	&0.0000 	&	0.0000 	&	0.0000 	&	0.0000 	&	0.0000 	&	0.0000\\
2008	&1.5802	    &0.0000 	&	4.7176 	&	4.8397 	&	0.8604 	&	1.6792 	&	0.0649\\
2009	&1.3244 	&2.4849 	&	0.9567 	&	0.8287 	&	1.5849 	&	1.6933 	&	1.1075\\
2010	&0.2096 	&0.1317 	&	1.3329 	&	1.1929 	&	0.2563 	&	0.2695 	&	0.1272\\
2011	&1.0253 	&1.7667 	&	0.6464 	&	0.4929 	&	1.0518 	&	1.7081 	&	1.1565\\
2012	&0.6564 	&0.9543 	&	2.0262 	&	1.8624 	&	0.6468 	&	0.8703 	&	0.7035\\
2013	&0.2137  	&0.6651 	&	0.1944 	&	0.0149 	&	0.1545 	&	2.4889 	&	0.1431\\
2014	&0.0371 	&0.7024 	&	1.1494 	&	1.3442 	&	0.1911 	&	2.6150 	&	0.0675\\[5pt]
2015	&0.4289 	&0.2417 	&	2.9130 	&	3.1242 	&	0.6852 	&	2.0399 	&	0.0599\\
2016	&0.7827 	&0.3679 	&	5.2414 	&	5.4707 	&	1.1359 	&	0.7595 	&	0.2976\\
2017	&1.9298 	&0.2931 	&	7.2081 	&	7.4553 	&	2.3620 	&	2.2365 	&	0.0722\\[5pt]
RMSEPR	&0.7210 	&0.9579 	&	1.5748 	&	1.5108 	&	0.6780 	&	1.6178 	&	\bf{\underline{0.4815}}\\
RMSEPO	&1.0471 	&0.3009 	&	5.1208 	&	5.3501 	&	1.3944 	&	1.6786 	&	\bf{\underline{0.1432}}\\
RMSE	&0.8188 	&0.7608 	&	2.6386 	&	2.6626 	&	0.8929 	&	1.6360 	&	\bf{\underline{0.3800}}\\
\hline
\end{tabular} }}
\end{table}
%%%%%%%%%%%%%%%%%%%%%%%%%%%%%%%%%%%%%%%%%%%%%%%%%%%%%%%%%%%%%%%%%%%%%%%%%%%%%%%
\begin{figure}[!htbp]
\centering
\includegraphics[height=6.2cm,width=17cm]{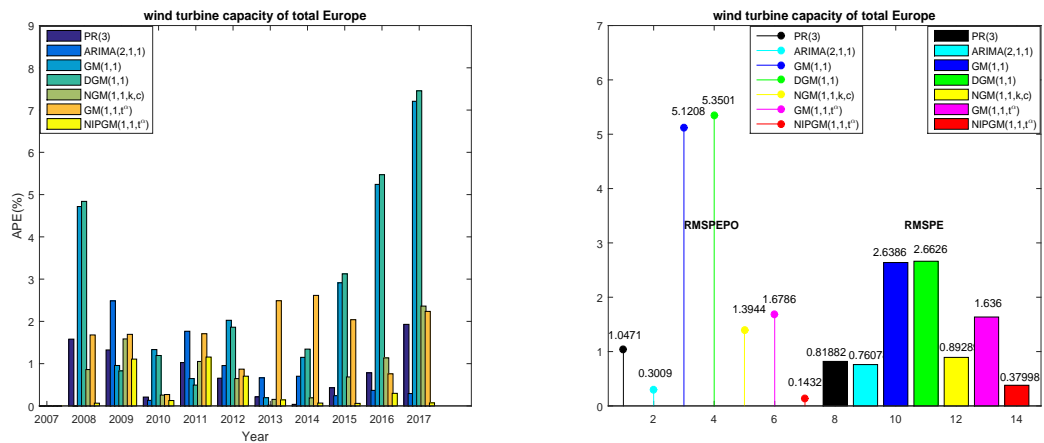}
\caption{Errors of different models in wind turbine capacity of Europe.}
 \label{fig:figure11}
\end{figure}
%%%%%%%%%%%%%%%%%%%%%%%%%%%%%%%%%%%%%%%%%%%%%%%%%%%%%%%%%%%%%%%%%%%%%%%
%=======================================================================
\newpage
\subsection{Total wind turbine capacity of North America}
 \label{subsec:North-America}
In this subsection, the total wind turbine capacity in North America is studied by seven forecasting models.
Similar to the analysis in section 4, Table \ref{table:table13} lists the minimum RMSE and the corresponding optimized parameters
of the two models.
%------------------------------------------------------------------
\begin{table}[!htbp]
\caption{The optimal parameters values of GM(1, 1, ${{t}^{\alpha }})$ model and the novel model in total wind turbine capacity of North America.}
 \label{table:table13}
 \vspace{6pt}
 \renewcommand{\baselinestretch}{1.25}
 {\footnotesize\centerline{\tabcolsep=9pt
 \begin{tabular}{ccccccc}
   \hline
     \multicolumn{3}{c}{NIPGM(1, 1, ${{t}^{\alpha }})$} && \multicolumn{2}{c}{GM(1, 1, ${{t}^{\alpha }})$} \\
    \hline
  RMSE(\%)      & $\lambda $      &$\alpha $      && RMSE(\%)    & $\alpha $      \\
 \hline
{{\bf 2.4236}}   &0.9086    &0.2637   && 5.4472 & 1.6400     \\
\hline
\end{tabular} }}
\end{table}

%%%%%%%%%%%%%%%%%%%%%%%%%%%%%%%%%%%%%%%%%%%%%%%%%%%%%%%%%%%%%%%%%%%%%%%%%%%%%%%%%%%%%
Further, the forecasting models are built by using the wind turbine capacity data from 2007 to 2014 in
North America. Table \ref{table:table14} gives the fitting and predicting values of the seven models.
It can be noticed from Fig. \ref{fig:preAmercian} and Table \ref{table:table14} that
the established seven forecasting models have an excellent predictive effect in wind turbine capacity forecasting of North America,
which can reflect the trend of wind turbine capacity from 2015 to 2017.
\begin{table}[!htbp]
\caption{The fitted and predicted values of different models in wind turbine capacity of North America(Megawatts). }
 \label{table:table14}
 \vspace{6pt}
 \renewcommand{\baselinestretch}{1.25}
 {\footnotesize\centerline{\tabcolsep=3pt
 \begin{tabular}{ccccccccccccc}
 \hline
Year  & data  &PR(1)    &ARIMA(2, 1, 1)  & GM(1, 1)	& DGM(1, 1)	 & NGM(1, 1, $k$, $c)$	 & GM(1, 1, ${{t}^{\alpha }})$	 & NIPGM(1, 1, ${{t}^{\alpha }})$\\
\hline
2007 	&18810.0000 &19869.0833	&18810.0000	&18810.0000 &18810.0000 &18810.0000 &18810.0000 &18810.0000\\
2008 	&27940.0000 &28534.6667	&27940.0000	&33430.1116 &33498.4981 &28055.3165 &23953.1486 &27265.0412\\
2009 	&38933.0000 &37200.2500	&37524.5787	&38866.9876 &38957.6409 &37830.7879 &34862.0700 &37921.7659\\
2010 	&45054.0000 &45865.8333	&48161.9090	&45188.0851 &45306.4426 &47064.7950 &45973.1476 &46970.6018\\
2011 	&53485.0000 &54531.4167	&54092.3164	&52537.2085 &52689.8882 &55787.3295 &56359.8353 &55366.0191\\
2012 	&67934.0000 &63197.0000	&62249.4276	&61081.5500 &61276.5902 &64026.7220 &65947.8132 &63448.7993\\
2013 	&71093.0000 &71862.5833	&76226.4435	&71015.4928 &71262.6395 &71809.7338 &74849.8077 &71379.1940\\[5pt]
2014 	&78340.0000 &80528.1667	&79298.1398	&82565.0334 &82876.0831 &79161.6438 &83193.0517 &79243.8894\\
2015 	&87058.4200 &89193.7500	&86307.0648	&95992.9231 &96382.1325 &86106.3307 &91081.8126 &87093.9468\\
2016 	&96994.0000 &97859.3333	&94743.4150	&111604.6455 &112089.2193 &92666.3508 &98595.3608 &94961.1729\\
2017 	&104070.0000 &106524.9167	&104358.1253	&129755.3663 &130356.0397 &98863.0108 &105793.3036 &102866.1383\\%[5pt]
%2018 	&		     &115190.5000	&111210.3022	&150858.0132 &151599.7453 &104716.4371 	&112721.0034 &110822.4755\\
%2019 	&		     &123856.0833	&118121.0296	&175392.6700 &176305.4693 &110245.6416 	&119413.6732 &118839.3441\\
%2020 	&		     &132521.6667	&124809.9313	&203917.4986 &205037.4059 &115468.5829 	&125899.2314 &126922.9203\\
\hline
\end{tabular} }}
\end{table}
%%%%%%%%%%%%%%%%%%%%%%%%%%%%%%%%%%%%%%%%%%%%%%%%%%%%%%%%%%%%%%%%%%%%%%%%%%%%%%%%%%%%
\begin{figure}[!htbp]
\centering
\includegraphics[height=7.8cm,width=12cm]{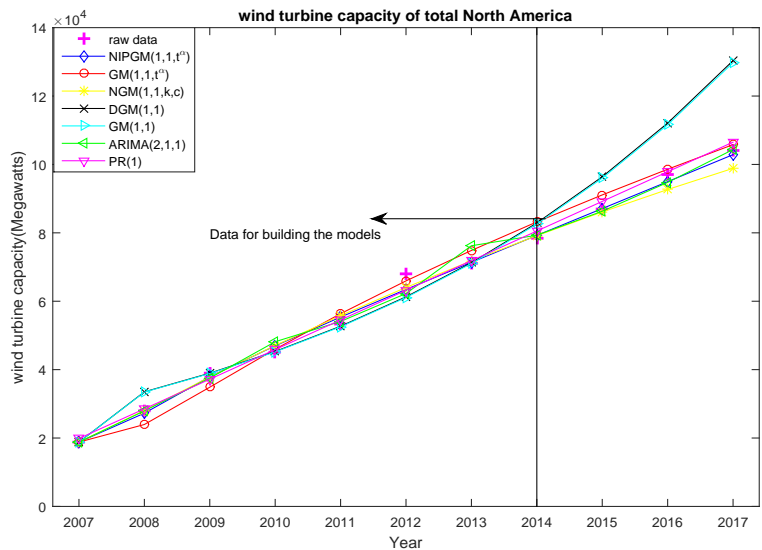}
\caption{The fitted and predicted values of different models in wind turbine capacity of North America.}
 \label{fig:preAmercian}
\end{figure}

%%%%%%%%%%%%%%%%%%%%%%%%%%%%%%%%%%%%%%%%%%%%%%%%%%%%%%%%%%%%%%%%%%%%%%%%%%%%%%%%%%%%%
According to the Eq.(\ref{Eq:num39}) to Eq.(\ref{Eq:num41}), the values of fitting and forecasting errors of the seven models can be calculated. Table \ref{table:table15} and Fig.\ref{fig:figure12}
show that the RMSEPR, RMSEPO and RMSE of NIPGM$(1,1,{{t}^{\alpha }})$ are 2.9918, 1.0978, 2.4236.
The numerical results expose that the forecasting and fitting accuracy of the novel model are higher
than other models. And new information priority accumulation has great significance for enhancing  the prediction
accuracy. Therefore, the proposed model is very suitable for forecasting wind turbine capacity in North America.
%%%%%%%%%%%%%%%%%%%%%%%%%%%%%%%%%%%%%%%%%%%%%%%%%%%%%%%%%%%%%%%%%%%%%%%%%%%%%%%%%%%%%
\begin{table}[!htbp]
\caption{Errors of different models in  wind turbine capacity of North America(\%). }
 \label{table:table15}
 \vspace{6pt}
 \renewcommand{\baselinestretch}{1.25}
 {\footnotesize\centerline{\tabcolsep=3pt
 \begin{tabular}{ccccccccccccc}
 \hline
 Year  &PR(1)    &ARIMA(2, 1, 1) & GM(1, 1)	& DGM(1, 1)	 & NGM(1, 1, $k$, $c)$	 & GM(1, 1, ${{t}^{\alpha }})$	 & NIPGM(1, 1, ${{t}^{\alpha }})$\\
\hline
2007 	&	5.6304 		&	0.0000 		&	0.0000 	&0.0000 	&	0.0000 	&	0.0000 	&	0.0000\\
2008 	&	2.1284 		&	0.0000 		&	19.6496 &19.8944 	&	0.4127 	&	14.2693 &	2.4157\\
2009 	&	4.4506 		&	3.6176 		&	0.1696 	&0.0633 	&	2.8310 	&	10.4562 &	2.5974\\
2010 	&	1.8019 		&	6.8982 		&	0.2976 	&0.5603 	&	4.4631 	&	2.0401 	&	4.2540\\
2011 	&	1.9565 		&	1.1355 		&	1.7721 	&1.4866 	&	4.3046 	&	5.3750 	&	3.5169\\
2012 	&	6.9729 		&	8.3678 		&	10.0869 &9.7998 	&	5.7516 	&	2.9237 	&	6.6023\\
2013 	&	1.0825 		&	7.2207 		&	0.1090 	&0.2386 	&	1.0082 	&	5.2844 	&	0.4026\\
2014 	&	2.7932 		&	1.2231 		&	5.3932 	&5.7903 	&	1.0488 	&	6.1949 	&	1.1538\\[5pt]
2015 	&	2.4528 		&	0.8630 		&	10.2627 &10.7097 	&	1.0936 	&	4.6215 	&	0.0408\\
2016 	&	0.8922 		&	2.3203 		&	15.0635 &15.5630 	&	4.4618 	&	1.6510 	&	2.0958\\
2017 	&	2.3589 		&	0.2769 		&	24.6809 &25.2580 	&	5.0034 	&	1.6559 	&	1.1568\\[5pt]
RMSEPR	&	3.0266 		&	4.0661 		&	5.3540 	&5.4048 	&	2.8314 	&	6.6491 	&	\bf{\underline{2.9918}}\\
RMSEPO	&	1.9013 		&	1.1534 		&	16.6690 &17.1769 	&	3.5196 	&	2.6428 	&	\bf{\underline{1.0978}}\\
RMSE	&	2.6890 		&	3.1923 		&	8.7485 	&8.9364 	&	3.0379 	&	5.4472 	&	\bf{\underline{2.4236}}\\
\hline
\end{tabular} }}
\end{table}
%%%%%%%%%%%%%%%%%%%%%%%%%%%%%%%%%%%%%%%%%%%%%%%%%%%%%%%%%%%%%%%%%%%%%%%%%%%%%%%
\begin{figure}[!htbp]
\centering
\includegraphics[height=5.7cm,width=17cm]{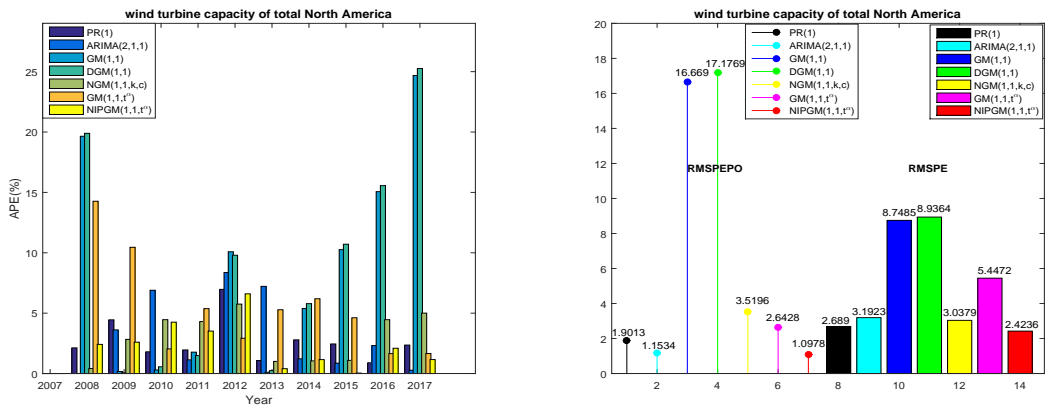}
\caption{Errors of different models in wind turbine capacity of North America.}
 \label{fig:figure12}
\end{figure}
%-------------------------------------------------------------------
%-------------------------------------------------------------------
\subsection{ Total wind turbine capacity of Asia}
 \label{subsec: Asia}
This subsection considers the total wind turbine capacity of Asia through seven forecasting models.
According to the discussion in Section 4, the minimum RMSE and corresponding optimal parameters of the two models are listed Table \ref{table:table16}.
 %------------------------------------------------------------------
\begin{table}[!htbp]
\caption{ The optimal parameters of GM(1,1,${{t}^{\alpha }})$ model and the novel model in total wind turbine capacity of Asia.}
 \label{table:table16}
 \vspace{6pt}
 \renewcommand{\baselinestretch}{1.25}
 {\footnotesize\centerline{\tabcolsep=9pt
 \begin{tabular}{ccccccc}
 \hline
     \multicolumn{3}{c}{NIPGM(1, 1, ${{t}^{\alpha }})$} && \multicolumn{2}{c}{GM(1, 1, ${{t}^{\alpha }})$} \\
    \hline
  RMSE(\%)      & $\lambda $      &$\alpha $      && RMSE(\%)    & $\alpha $      \\
 \hline
{{\bf 3.3256}}   &0.9014   &1.0978   && 4.5568 &2.1840     \\
\hline
\end{tabular} }}
\end{table}

\begin{table}[!htbp]
\caption{The fitted and predicted values of different models in wind turbine capacity of Asia(Megawatts). }
 \label{table:table17}
 \vspace{6pt}
 \renewcommand{\baselinestretch}{1.25}
 {\footnotesize\centerline{\tabcolsep=3pt
 \begin{tabular}{ccccccccccccc}
 \hline
 Year  & data  &PR(2)   &ARIMA(2, 1, 2)  & GM(1, 1)	& DGM(1, 1)	 & NGM(1, 1, $k$, $c)$	 & GM(1, 1, ${{t}^{\alpha }})$	 & NIPGM(1, 1, ${{t}^{\alpha }})$\\
\hline
2007 &15327.3260 &13614.4778	&15327.3260	&15327.3260 &15327.3260 	&15327.3260 	&15327.3260 	&15327.3260\\
2008 &22356.3570 &23523.2229	&24442.2317	&30411.7783 &30593.2156 	&20761.8645 	&19571.5731 	&21279.7344\\
2009 &33737.5070 &35665.8868	&36990.6742	&39064.6946 &39349.6869 	&34761.6976 	&32086.3194 	&34904.5823\\
2010 &48622.3270 &50042.4692	&44795.2174	&50179.5833 &50612.4585 	&50156.2873 	&49814.9600 	&50316.4550\\
2011 &69073.8140 &66652.9702	&66555.9999	   &64456.9375 &65098.8905 	&67084.5888 	&69887.5839 	&67445.5987\\
2012 &87572.6850 &85497.3899	&85642.7395	    &82796.5584 &83731.6674 	&85699.4005 	&91305.4153 	&86273.5270\\
2013 &105496.3320 &106575.7282  &109051.7419	&106354.2630 &107697.5672 	&106168.7439 	&113672.2917 	&106801.6654\\
2014 &129273.7820 &129887.9851  &124833.1761	&136614.7274 &138523.0502 	&128677.3797 &136798.4924 	&129041.2131\\[5pt]
2015 &167528.3270 &155434.1606  &156064.2433	&175485.0554 &178171.4847 	&153428.4759 	&160573.3168 	&153008.8500\\
2016 &189684.6370 &183214.2548  &194671.4654	&225414.9700 &229168.1992 	&180645.4414 	&184921.4951 	&178724.6215\\
2017 &209977.2340 &213228.2675  &217864.4953    &289551.2019 &294761.3285 &210573.9425 &209786.6733 	&206210.7892\\%[5pt]
%2018 &		&245476.1989	&241618.0783	   &371935.8057 	 &379128.6970 	&243484.1204 &235124.2074 	&235491.1594\\
%2019 &		&279958.0489	&276122.0911	&477760.9028 	 &487643.9172 	&279673.0296 	&260897.4886 	&266590.6714\\
%2020 &		&316673.8175	&314273.2023	&613695.9031 	 &627218.6514 	&319467.3193 	&287075.7938 	&299535.1338\\
\hline
\end{tabular} }}
\end{table}
%%%%%%%%%%%%%%%%%%%%%%%%%%%%%%%%%%%%%%%%%%%%%%%%%%%%%%%%%%%%%%%%%%%%%%%%%%%%%%%%%%%%
\begin{figure}[!htbp]
\centering
\includegraphics[height=8cm,width=12cm]{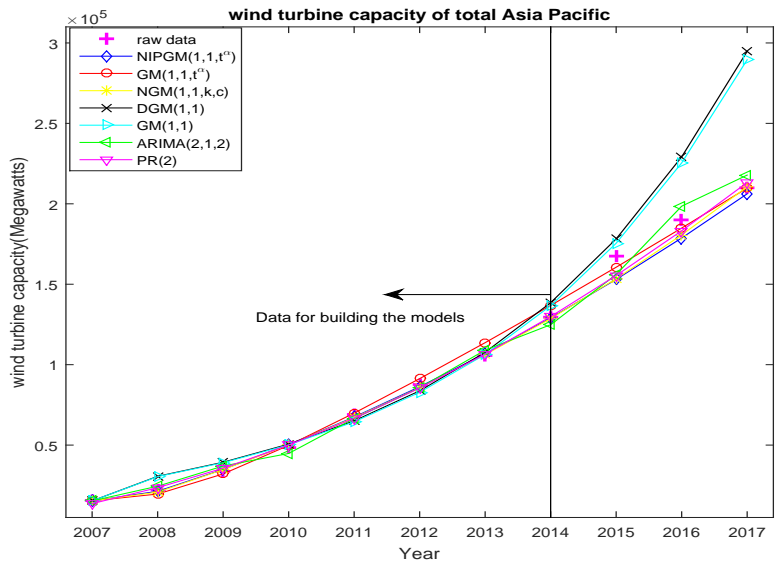}
\caption{The fitted and predicted values of different models in wind turbine capacity of Asia.}
 \label{fig:preAsia}
\end{figure}
%%%%%%%%%%%%%%%%%%%%%%%%%%%%%%%%%%%%%%%%%%%%%%%%%%%%%%%%%%%%%%%%%%%%%%%%%%%%%%%%%%%%%
%%%%%%%%%%%%%%%%%%%%%%%%%%%%%%%%%%%%%%%%%%%%%%%%%%%%%%%%%%%%%%%%%%%%%%%%%%%%%%%%%%%%%
\begin{table}[!htbp]
\caption{Errors of different models in  wind turbine capacity of Asia(\%). }
 \label{table:table18}
 \vspace{6pt}
 \renewcommand{\baselinestretch}{1.25}
 {\footnotesize\centerline{\tabcolsep=3pt
 \begin{tabular}{ccccccccccccc}
 \hline
  Year  &PR(2)    &ARIMA(2, 1, 2)  & GM(1, 1)	& DGM(1, 1)	 & NGM(1, 1, $k$, $c)$	 & GM(1, 1, ${{t}^{\alpha }})$	 & NIPGM(1, 1, ${{t}^{\alpha }})$\\
\hline
2007 	&11.1751 	&	0.0000 		&	0.0000 		&	0.0000 	&	0.0000 	&	0.0000 	&	0.0000\\
2008 	&5.2194 	&	9.3301 		&	36.0319 	&	36.8435 &	7.1322 	&	12.4563 &	4.8157\\
2009 	&5.7158 	&	9.6426   	&	15.7901 	&	16.6348 &	3.0358 	&	4.8942 	&	3.4593\\
2010 	&2.9208 	&	7.8711 		&	3.2028 		&	4.0930 	&	3.1548 	&	2.4529 	&	3.4843\\
2011 	&3.5047 	&	3.6451 		&	6.6840 		&	5.7546 	&	2.8799 	&	1.1781 	&	2.3572\\
2012 	&2.3698 	&	2.2038 		&	5.4539 		&	4.3861 	&	2.1391 	&	4.2624 	&	1.4835\\
2013 	&1.0232 	&	3.3702 		&	0.8132 		&	2.0866 	&	0.6374 	&	7.7500 	&	1.2373\\
2014 	&0.4751 	&	3.4350 		&	5.6786 		&	7.1548 	&	0.4613 	&	5.8208 	&	0.1799\\[5pt]
2015 	&7.2192 	&	6.8431 		&	4.7495 		&	6.3530 	&	8.4164 	&	4.1515 	&	8.6669\\
2016 	&3.4111 	&	4.6528 		&	18.8367 	&	20.8154 &	4.7654 	&	2.5111 	&	5.7780\\
2017 	&1.5483 	&	3.7562 		&	37.8965 	&	40.3778 &	0.2842 	&	0.0908 	&	1.7937\\[5pt]
RMSEPR	&3.0327 	&	5.6426 		&	10.5221 	&	10.9933 &	2.7772 	&	5.5450 	&	\bf{\underline{2.4310}}\\
RMSEPO	&4.0595 	&	5.0841 		&	20.4942 	&	22.5154 &	4.4887 	&	\bf{\underline{2.2511}} 	&	5.4129\\
RMSE	&3.3407 	&	5.4750 		&	13.5137 	&	14.4500 &\bf{\underline{3.2906}} 	&	4.5568 	&	3.3256\\
\hline
\end{tabular} }}
\end{table}
%%%%%%%%%%%%%%%%%%%%%%%%%%%%%%%%%%%%%%%%%%%%%%%%%%%%%%%%%%%%%%%%%%%%%%%%%%%%%%%
Table \ref{table:table17} and Fig. \ref{fig:preAsia} display the simulation and prediction values. Table \ref{table:table18} and
Fig. \ref{fig:figure13} show the results of APE, RMSEPR, RMSEPO, and RMSE. We can notice from Fig.\ref{fig:figure13} and Table \ref{table:table18} that the maximum forecasting error of the basic grey model and discrete grey model are as high as 37.8965\%,40.3778\%, respectively. It means that the forecasting error of the basic grey model and discrete grey model are enormous. Further, the RMSEPR of the novel
model is 2.4310, the RMSEPO of the grey model with time power is 2.2511, and the RMSE of the extended non-homogeneous exponential grey model is 3.2906. The results display that the prediction errors of the novel grey model are smaller than other models. Besides, the grey model with time power and extended non-homogeneous exponential grey model are particular cases of the proposed model. What's more, it can indicate that the new information prioritization is important for forecasting.

\begin{figure}[!htbp]
\centering
\includegraphics[height=5.7cm,width=17cm]{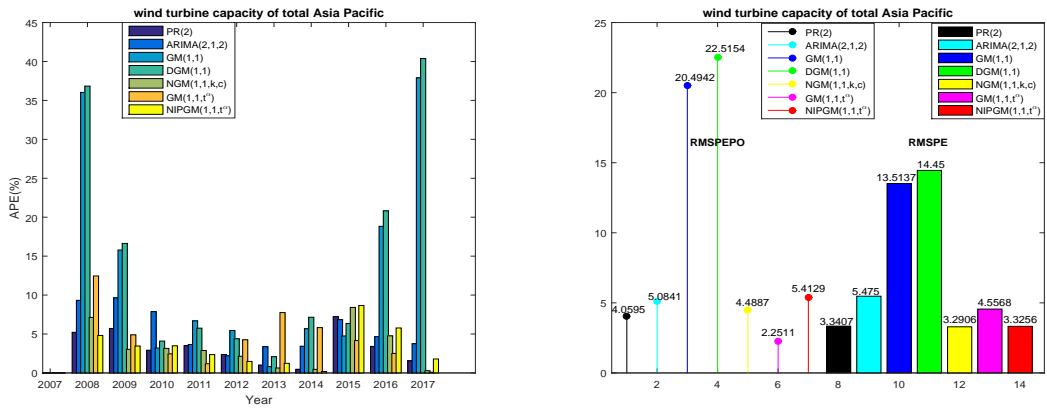}
\caption{Errors of different models in wind turbine capacity of Asia.}
 \label{fig:figure13}
\end{figure}

%-------------------------------------------------------------------
%-------------------------------------------------------------------
\subsection{ Total wind turbine capacity of the world}
 \label{subsec: world}
 In this subsection, the global wind turbine capacity is derived through seven prediction models.
 Table \ref{table:table19} lists the minimum RMSE and the corresponding optimal values of the two models.
 %------------------------------------------------------------------
\begin{table}[!htbp]
\caption{ The optimal parameters values of GM(1, 1, ${{t}^{\alpha }})$ model and the novel model in total wind turbine capacity of the world.}
 \label{table:table19}
 \vspace{6pt}
 \renewcommand{\baselinestretch}{1.25}
 {\footnotesize\centerline{\tabcolsep=9pt
 \begin{tabular}{ccccccc}
  \hline
     \multicolumn{3}{c}{NIPGM(1, 1, ${{t}^{\alpha }})$} && \multicolumn{2}{c}{GM(1, 1, ${{t}^{\alpha }})$} \\
    \hline
  RMSE(\%)      & $\lambda $      &$\alpha $      && RMSE(\%)    & $\alpha $      \\
 \hline
{{\bf 1.4391}}   &0.7161   &1.3276   && 4.0898 &2.8605     \\
\hline
\end{tabular} }}
\end{table}

\begin{table}[!htbp]
\caption{The fitted and predicted values of different models in wind turbine capacity of the world(Megawatts).}
 \label{table:table20}
 \vspace{6pt}
 \renewcommand{\baselinestretch}{1.25}
 {\footnotesize\centerline{\tabcolsep=3pt
 \begin{tabular}{ccccccccccccc}
 \hline
 Year  & data  &PR(2)    &ARIMA(2, 1, 1)  & GM(1, 1)	& DGM(1, 1)	 & NGM(1, 1, $k$, $c)$	 & GM(1, 1, ${{t}^{\alpha }})$	 & NIPGM(1, 1, ${{t}^{\alpha }})$\\
\hline
2007 &91894.0080  &89895.7940  &91894.0080	&91894.0080  &91894.0080  &91894.0080  &91894.0080  &91894.0080\\
2008 &116511.6230 &119032.8856 &116511.6230	&129789.1184 &130130.5184 &115887.4467 	&118956.8735 &117442.6582\\
2009 &151655.8934 &150938.8752 &152167.2895	&153837.9115 &154305.9948 &149773.8023 	&148814.4211 &151376.2996\\
2010 &182901.3012 &185613.7630 &188450.4898	&182342.7364 &182972.7593 &185600.8220 	&184697.3033 &187037.5395\\
2011 &222516.8618 &223057.5489 &221338.5422	&216129.2570 &216965.1975 &223479.6474 	&226537.3484 &224763.4997\\
2012 &269853.3279 &263270.2329 &262789.1973	&256176.1257 &257272.7062 &263527.7848 	&273643.5552 &264896.4706\\
2013 &303112.5198 &306251.8150 &313238.8541	&303643.3302 &305068.4908 &305869.4703 	&324527.1561 &307861.1290\\
2014 &351617.6747 &352002.2953 &348668.0199	&359905.7943 &361743.7133 &350636.0549 	&376684.5272 &354209.7165\\[5pt]
2015 &417144.1127 &400521.6737 &399131.1492	&426593.2028 &428947.9839 &397966.4121 	&426326.0794 &404666.5259\\
2016 &467698.4951 &451809.9502 &469387.1588	&505637.2072 &508637.3755 &448007.3685 	&468036.6722 &460180.8872\\
2017 &514798.1313 &505867.1249 &524331.1065	&599327.3771 &603131.3573 &500914.1595 	&494349.5839 &521994.0838\\%[5pt]
%2018 &			  &562693.1977 &573530.8348	&710377.5193 &715180.3066 &556850.9109 	&495211.5578 &591725.3917\\
%2019 &			  &622288.1686 &634840.9951	&842004.2856 &848045.5621 &615991.1479 	&457310.7263 &671483.3296\\
%2020 &			  &684652.0376 &699753.3108	&998020.3450 &1005594.3498&678518.3333 	&363231.9909 &764009.7685\\
\hline
\end{tabular} }}
\end{table}

Furthermore, all the results of seven forecasting models are listed in Table \ref{table:table20}, Fig.\ref{fig:preWorld}, Table \ref{table:table21} and Fig.\ref{fig:figure14}.
 We can observe in Table \ref{table:table21}, and Fig.\ref{fig:figure14}, the RMSEPR of the extended non-homogeneous exponential grey model is 1.0312.
 However, it is worth noting that the RMSEPO and RMSE of
 NIPGM$(1, 1,{{t}^{\alpha }})$ are 1.9988, 1.4392, respectively. The results display that the extended non-homogeneous exponential grey model
 has the best accuracy of fitting and the NIPGM$(1,1,{{t}^{\alpha }})$ has the best performance of forecasting. The reason is that the extended non-homogeneous exponential grey model is a particular case of the proposed model. Thus, the new grey model is suitable for predicting wind turbine capacity in the world.
%%%%%%%%%%%%%%%%%%%%%%%%%%%%%%%%%%%%%%%%%%%%%%%%%%%%%%%%%%%%%%%%%%%%%%%%%%
\begin{figure}[!htbp]
\centering
\includegraphics[height=8cm,width=12cm]{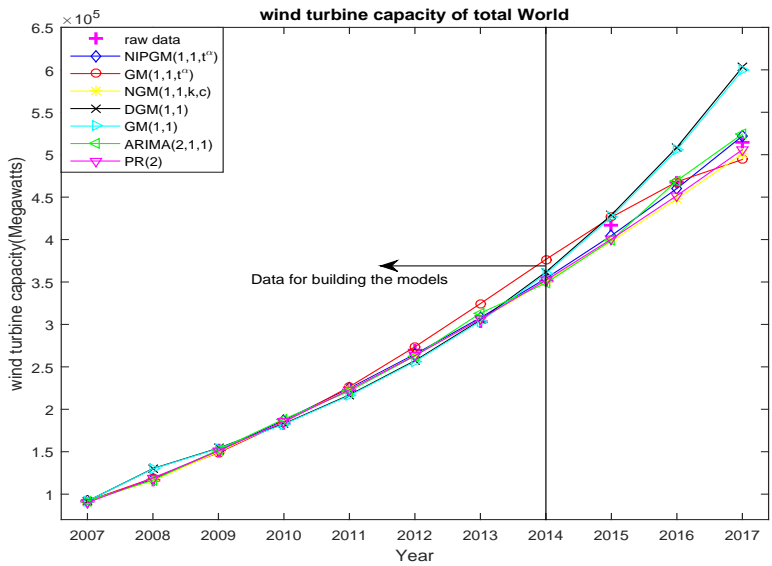}
\caption{The fitted and predicted values of different models in global wind turbine capacity.}
 \label{fig:preWorld}
\end{figure}
%%%%%%%%%%%%%%%%%%%%%%%%%%%%%%%%%%%%%%%%%%%%%%%%%%%%%%%%%%%%%%%%%%%%%%%%%%%%%%%%%%%%%
%%%%%%%%%%%%%%%%%%%%%%%%%%%%%%%%%%%%%%%%%%%%%%%%%%%%%%%%%%%%%%%%%%%%%%%%%%%%%%%%%%%%%
\begin{table}[!htbp]
\caption{Errors of different models in global wind turbine capacity(\%). }
 \label{table:table21}
 \vspace{6pt}
 \renewcommand{\baselinestretch}{1.25}
 {\footnotesize\centerline{\tabcolsep=6pt
 \begin{tabular}{ccccccccccccc}
 \hline
  Year  &PR(2)    &ARIMA(2, 1, 1)  & GM(1, 1)	& DGM(1, 1)	 & NGM(1, 1, $k$, $c)$	 & GM(1, 1, ${{t}^{\alpha }})$	 & NIPGM(1, 1, ${{t}^{\alpha }})$\\
\hline
2007 	 &2.1745	 &	0.0000 		&	0.0000 	&		0.0000 	&	0.0000 	&	0.0000 	&	0.0000\\
2008 	 &2.1640	 &	0.0000 		&	11.3959 &		11.6889 &	0.5357 	&	2.0987 	&	0.7991\\
2009 	 &0.4728	 &	0.3372 		&	1.4388 	&		1.7474 	&	1.2410 	&	1.8736 	&	0.1844\\
2010 	 &1.4830	 &	3.0340 		&	0.3054 	&		0.0391 	&	1.4759 	&	0.9820 	&	2.2615\\
2011 	 &0.2430	 &	0.5295 		&	2.8706 	&		2.4949 	&	0.4327 	&	1.8068 	&	1.0096\\
2012 	 &2.4395	 &	2.6178 		&	5.0684 	&		4.6620 	&	2.3441 	&	1.4046 	&	1.8369\\
2013 	 &1.0357	 &	3.3408 		&	0.1751 	&		0.6453 	&	0.9095 	&	7.0649 	&	1.5666\\
2014 	 &0.1094	 &	0.8389 		&	2.3571 	&		2.8798 	&	0.2792 	&	7.1290 	&	0.7372\\[5pt]
2015 	 &3.9848	 &	4.3182 		&	2.2652 	&		2.8297 	&	4.5974 	&	2.2011 	&	2.9912\\
2016 	 &3.3972	 &	0.3611 		&	8.1118 	&		8.7533 	&	4.2102 	&	0.0723 	&	1.6074\\
2017 	 &1.7349	 &	1.8518 		&	16.4199 &		17.1588 &	2.6970 	&	3.9721 	&	1.3978\\[5pt]
RMSEPR	 &1.1353	 &	1.5283 		&	3.3730 	&		3.4511 	&\bf{\underline{1.0312}} 	&	3.1942 	&	1.1993\\
RMSEPO	 &3.0390	 &	2.1770 		&	8.9323 	&		9.5806 	&	3.8349 	&	2.0819 	&	\bf{\underline{1.9988}}\\
RMSE	 &1.7064	 &	1.7229 		&	5.0408 	&		5.2899 	&	1.8723 	&	2.8605 	&	\bf{\underline{1.4392}}\\
\hline
\end{tabular} }}
\end{table}
%%%%%%%%%%%%%%%%%%%%%%%%%%%%%%%%%%%%%%%%%%%%%%%%%%%%%%%%%%%%%%%%%%%%%%%%%%%%%%%
\begin{figure}[!htbp]
\centering
\includegraphics[height=5.7cm,width=17cm]{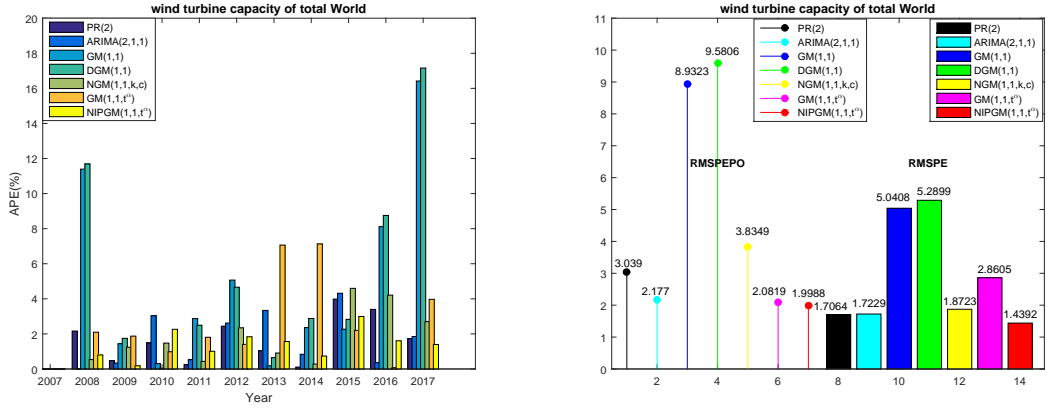}
\caption{Errors of different models in wind turbine capacity of the world.}
 \label{fig:figure14}
\end{figure}
%-----------------------------------------------------------------------------------
%--------------------------------------------------------------------

\subsection{Comparison of prediction models}
According to the results of forecasting wind turbine capacity in Europe, North America, Asia, and the world, we further compare the performance of accuracy of seven forecasting models.
The ranks of the performance of simulation and prediction are listed in Table  \ref{table:rankwind}.
As can be observed from Table \ref{table:rankwind}, all the prediction models can get
acceptable results by 8 years of wind turbine capacity data. However, it is worth noting that the NIPGM$(1,1,{{t}^{\alpha }})$ model has the best
simulation and prediction capability in wind turbine capacity forecasting.
\begin{table}[!htbp]
\caption{The average values and ranks of RMSEPR, RMSEPO and RMSE for seven models in wind turbine capacity forecasting. }
 \label{table:rankwind}
 \vspace{6pt}
 \renewcommand{\baselinestretch}{1.25}
 {\footnotesize\centerline{\tabcolsep=4pt
 \begin{tabular}{lcccccccccccc}
 \hline
                &PR($n$)     &ARIMA($p$, $d$, $q$)   & GM(1, 1)	& DGM(1, 1)	 & NGM(1, 1, $k$, $c)$	 & GM(1, 1, ${{t}^{\alpha }})$	 & NIPGM(1, 1, ${{t}^{\alpha }})$\\
\hline
Average RMSEPR     &2.6385	 &4.0650	&6.9413	&7.1200	&2.4393	&5.6687	&{\bf \underline{  2.3679  }}\\
Simulation rank   &3	  &4	  & 6	 &  7	 &  2	 &  5	&  {\bf \underline{      1       }}\\
Average RMSEPO    &3.3490	 &2.9051	&17.0721	&18.2077	&4.4125	&2.8848	&{\bf \underline{2.8842}}\\
Prediction rank   &4	  &3       & 6	  & 7	  & 5     & 2	  &   {\bf \underline{   1    }}\\
Average RMSE      &2.8517	 &3.7170	&9.9805	&10.4463	&3.03123	&4.8335	&{\bf \underline{  2.5228  }}\\
Overall rank      &2	      &4	& 6	      & 7	  & 3     & 5	       & {\bf \underline{     1    }}\\		
\hline
\end{tabular} }}
\end{table}

Among the seven prediction models,  both of the polynomial regression model and time series model require a lot of historical data when establishing the prediction model. In wind turbine capacity forecasting, using only 8 data to develop prediction models is the reason of poor prediction accuracy.
Compared with several grey models, we find that new information priority accumulation can significantly exploit data with new characteristic behaviors. Therefore, the NIPGM$(1, 1,{{t}^{\alpha }})$ model has a great advantage for small samples with new characteristic behaviors.

\subsection{Future discussion and development suggestion}
 \label{subsec:discussion}
As can be seen from the above discussion, the NIPGM$(1,1,{{t}^{\alpha }})$ outperforms other six prediction models.
Therefore, the NIPGM$(1,1,{{t}^{\alpha }})$ model will be applied to predict wind turbine capacity of Europe,
North America, Asia, and the world from 2018 to 2020. Table \ref{table:table22} and Fig.\ref{fig:figure15} lists the prediction values, and Table \ref{table:table23}
displays the annual increase rates.

 With the continuous deepening understanding of environmental issues around the world, and the constant improvement of renewable energy comprehensive utilization technologies, the global wind power industry has quickly improved in recent years.
At present, wind power generation has turned into one of the fastest-growing renewable energy sources, and its proportion has been increasing in clean energy production. Thus, it owes broad prospects for development.
Further, we will present the growth trend of the next three years of wind turbine capacity in Europe, North America, Asia, and the world.
%%%%%%%%%%%%%%%%%%%%%%%%%%%%%%%%%%%%%%%%%%%%%%%%%%%%%%%%%%%%%%%%%%%%%%%%%%%%%%%%%%%%%
\begin{table}[!htbp]
\caption{ The results of the wind turbine capacity in Europe, North America, Asia, and the world from 2017 to 2020(Megawatts). }
 \label{table:table22}
 \vspace{6pt}
 \renewcommand{\baselinestretch}{1.25}
 {\footnotesize\centerline{\tabcolsep=6pt
 \begin{tabular}{ccccccccccccc}
 \hline
 %\multirow{2}{*}{Year}   &  \multicolumn{4}{c}{wind turbine capacity(Megawatts)}\\
 %  \cline{2-5}
 Year                & Europe	&North America	 &Asia	 & World\\
\hline
2017	&178314.1463 	&104070.0000 	&209977.2340 	&514798.1313\\
2018	&195111.4074 	&110822.4755 	&235491.1594 	&591725.3917\\
2019	&213297.6361 	&118839.3441 	&266590.6714 	&671483.3296\\
2020	&232847.4225 	&126922.9203 	&299535.1338 	&764009.7685\\
\hline
\end{tabular} }}
\end{table}
%%%%%%%%%%%%%%%%%%%%%%%%%%%%%%%%%%%%%%%%%%%%%%%%%%%%%%%%%%%%%%%%%%%%%%%%%%%%%%%
%%%%%%%%%%%%%%%%%%%%%%%%%%%%%%%%%%%%%%%%%%%%%%%%%%%%%%%%%%%%%%%%%%%%%%%%%%%%%%%%%%%%%
\begin{table}[!htbp]
\caption{ The results of annual increase rate of wind turbine capacity in Europe, North America, Asia, and the world from 2017 to 2020(\%). }
 \label{table:table23}
 \vspace{6pt}
 \renewcommand{\baselinestretch}{1.25}
 {\footnotesize\centerline{\tabcolsep=6pt
 \begin{tabular}{ccccccccccccc}
 \hline
  Year                     & Europe	&North America	 &Asia	 & World\\
\hline
2017	 &10.1113 	 &7.2953 	 &10.6981 	 &10.0705\\
2018	 &9.4200 	 &6.4884 	 &12.1508 	 &14.9432\\
2019	 &9.3209 	 &7.2340 	 &13.2062 	 &13.4789\\
2020	 &9.1655 	 &6.8021 	 &12.3577 	 &13.7794\\
mean value	 &9.5045 	 &6.9549 	 &12.1032 	 &13.0680\\
\hline
\end{tabular} }}
\end{table}
%%%%%%%%%%%%%%%%%%%%%%%%%%%%%%%%%%%%%%%%%%%%%%%%%%%%%%%%%%%%%%%%%%%%%%%%%%%%%%%%%%%%%%
%%%%%%%%%%%%%%%%%%%%%%%%%%%%%%%%%%%%%%%%%%%%%%%%%%%%%%%%%%%%%%%%%%%%%%%%%%%%%%%
\begin{figure}[!htbp]
\centering
\includegraphics[height=12cm,width=16.5cm]{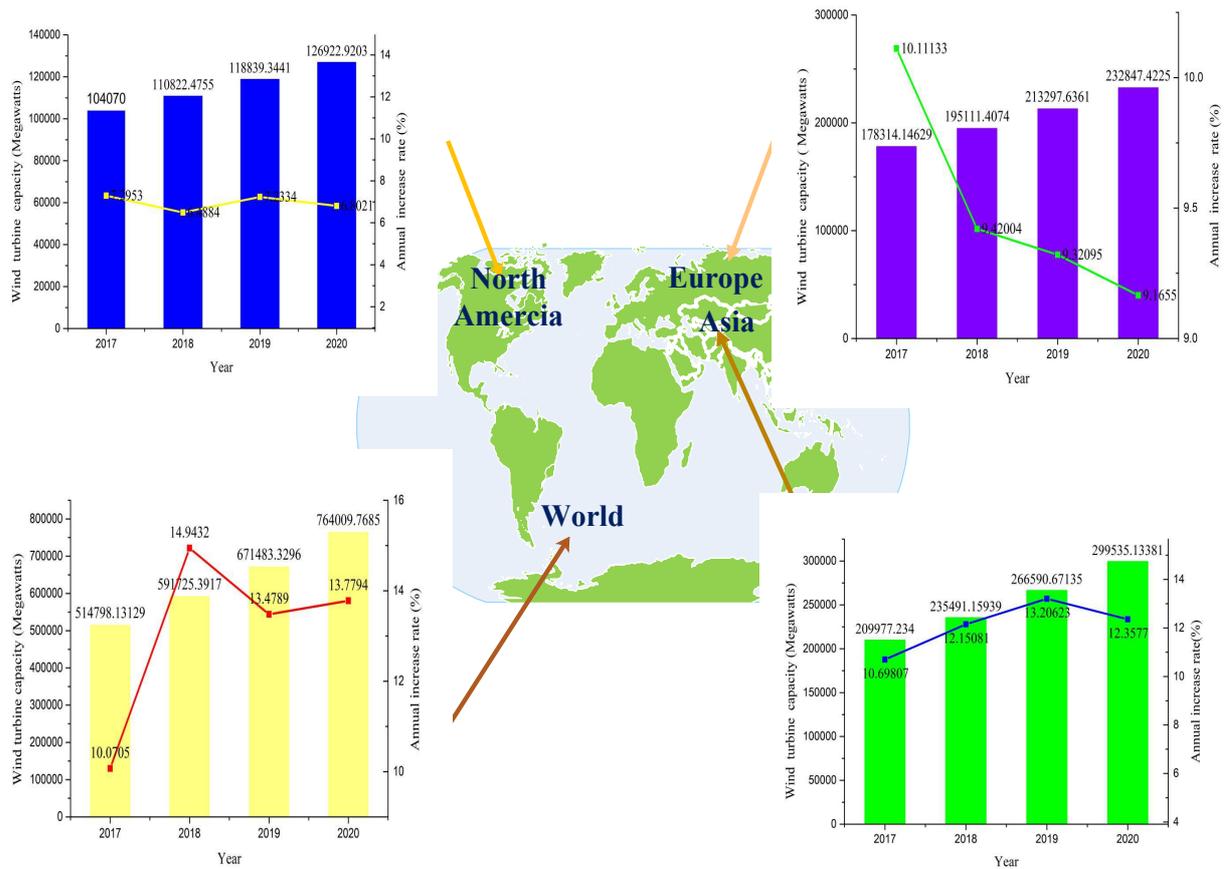}
\caption{The results of prediction and annual increase rate of wind turbine capacity in Asia ,Europe, North America, and the world from 2017 to 2020.}
 \label{fig:figure15}
\end{figure}

In Europe, it is predicted that the wind turbine capacity will maintain steady growth at an average annual growth rate of 9.5045\% approximately,
which will reach 232847.4225 Megawatts in 2020. And the annual increase rate of wind turbine capacity will decline in the next
three years. The result is consistent with the predicted result of the {\it GWEC Global Wind Report 2018} (https://gwec.net) and literature\citep{KLESSMANN20117637}. The reason is that many Member States of US reduce wind turbine capacity investments and support programs.
While the wind turbine capacity in Europe is increasing, which will face regional imbalances and bring challenges in materials, control, and storage\citep{SCARLAT2015969}. Therefore, it is possible to promote the excellent and rapid advancement of wind power generation in Europe by researching and developing diversified technologies for cleaner production.

In North America, the wind turbine capacity will grow at an average annual increase rate of 6.9549\%, and it will increase from 104070
Megawatts in 2017 to 126922.9303 Megawatts in 2020. It reflects that the wind turbine in North America has an unstable growth trend in the future.
Although North America has enormous potential for wind energy development, some parts still rely
on fossil fuels for power generation\citep{Mercer2017}.
The {\it GWEC Global Wind Report 2018} states that governments' commitment to large-scale auction is driving the wind turbine capacity.
Therefore, governments should maintain their commitment and make reasonable policies to promote positive development.

In Asia, the wind turbine capacity  will continue to grow at an average annual increase rate of 12.1032\% approximately, which is expected to reach 299535.1338 Megawatts in 2020. As we all know, Asia is the region with the largest installed capacity of global wind turbines.
However, the wind energy of South-East Asia and India will remain at a moderate level. Therefore,
South-East Asian governments should call for stop prioritizing coal\citep{SHUKLA2017342}.
For example, Vietnam should increase a higher tax on fossil fuel to promote the development of clean energy department\citep{NONG201990}.
India should drive the volume of wind turbine capacity with the execution of the scheduled auctions\citep{Neeru2019Renewable}.
Furthermore, the relevant practitioners can use effective wind energy prediction methods to increase wind turbine capacity\citep{Yang2019hybrid}.
Besides, Asian governments should correctly recognize the opportunities and challenges for renewable energy and
increase cooperation among countries of the renewable energy sector\citep{SHARVINI2018257}.

In the world, the wind turbine capacity will maintain sharply increase from 514798.1313 Megawatts in 2017 to
764009.7685 Megawatts in 2020, which the average annual growth rate is 13.0068\%.
Currently, global wind power accounts for 16\% of renewable energy\citep{HEUBERGER2018367}.
With the rapid development of wind energy, wind power generation has received more and more attention from all over the world.
It is imperative to increase the construction of wind turbines, expand the scale of clean energy supply\citep{LUNDIE20191042}, and promote the continuous improvement of production technology by power generation enterprises.
Countries around the world should increase communication and cooperation\citep{AKIZUGARDOKI20181145} between the wind power sectors.
These suggestions can contribute to advance of the low carbon economy and achieve rapid development of clean energy production.
%%%%%%%%%%%%%%%%%%%%%%%%%%%%%%%%%%%%%%%%%%%%%%%%%%%%%%%%%%%%%%%%%%%%%
%
\section{Conclusions}
\label{sec:conclu}

In this paper, combining the new information priority accumulation with grey GM(1, 1, ${{t}^{\alpha }})$ model, we propose a novel NIPGM$(1, 1,{{t}^{\alpha }})$  model to predict short-term wind turbine capacity of Europe, North America, Asia, and the world.
The NIPGM$(1, 1,{t}^{\alpha })$ model is a more generic model. The traditional GM$(1, 1)$ model, NGM(1, 1, $k$, $c)$ model, NGM(1, 1, $k)$ model, and GM(1, 1, ${t}^{\alpha })$ model are special cases of the proposed model with determined parameters $\lambda $ and $\alpha$.

Three numerical cases are used to evaluate the accuracy of the novel model and six existing prediction models. It shows that the NIPGM$(1, 1,{t}^{\alpha })$ model has a sufficient advantage for small samples than the polynomial regression model and the time series model. Furthermore, the proposed model is superior to other forecasting models, the results reveal that new information priority accumulation is an effective method to improve the prediction ability of grey model. Besides, the particle swarm optimization algorithm is very stable when determining the optimal values of  nonlinear parameters.

The proposed model is applied to predict the total wind turbine capacity, and it has the highest simulation and prediction accuracy than other
commonly prediction models.
 It is predicted that the average annual increase rate of the total wind turbine capacity in Europe, North America, Asia, and the world from 2018 to 2020 are 9.5045\%, 6.9549\%, 12.1032\%, 13.0680\%, respectively.
In the future, wind energy will make an enormous contribution to the sustainable development of cleaner production.
Additionally, reasonable suggestions are put forward from the standpoint of the practitioners and governments.
(1)The related practitioners can use useful prediction models to predict wind turbine capacity to ensure the installation of wind turbines.
 And they can strengthen the construction of the wind power system and innovate wind power technology to achieve an efficient configuration
 of the power supply system and develop economic efficiency.
(2)Governments around the world should increase communication and cooperation to promote competition and development in the global energy market.  Moreover, governments can make rational policies and increase the cost-effectiveness of renewable energy into the power supply, which can promote the sustainable development of clean energy production in the future.

From the perspective of the new information accumulation, it focuses on mining new information rules, and new information has a more significant impact on the forecasting. Therefore, the grey forecasting model with new information priority accumulation is
suitable for the small sample with new characteristic behaviors. It can be applied to the prediction of other
cleaner production such as solar and natural gas. In the future, we will combine the new information priority accumulation with
other grey models to research whether the prediction accuracy of other grey models would be improved.

%%%%%%%%%%%%%%%%%%%%%%%%%%%%%%%%%%%%%%%%%%%%%%%%%%%%%%%%%%%%%%%%%%%%%

\section*{Acknowledgments}

This paper was supported by the National Natural Science
Foundation of China (No.71901184, 71771033, 71571157, 11601357), the
Humanities and Social Science Project of Ministry of Education of
China (No.19YJCZH119), the funding of V.C. \& V.R. Key Lab
of Sichuan Province (SCVCVR2018.10VS), National Statistical Scientific Research Project (2018LY42),
the Open Fund (PLN201710) of State Key Laboratory of Oil and Gas Reservoir Geology and
Exploitation, and the Longshan academic talent research
supporting program of SWUST (No.17LZXY20).

%%%%%%%%%%%%%%%%%%%%%%%%%%%%%%%%%%%%%%%%%%%%%%%%%%%%%%%%
\section*{References}

\bibliography{greybib}

\begin{thebibliography}{65}
\expandafter\ifx\csname natexlab\endcsname\relax\def\natexlab#1{#1}\fi
\providecommand{\url}[1]{\texttt{#1}}
\providecommand{\href}[2]{#2}
\providecommand{\path}[1]{#1}
\providecommand{\DOIprefix}{doi:}
\providecommand{\ArXivprefix}{arXiv:}
\providecommand{\URLprefix}{URL: }
\providecommand{\Pubmedprefix}{pmid:}
\providecommand{\doi}[1]{\href{http://dx.doi.org/#1}{\path{#1}}}
\providecommand{\Pubmed}[1]{\href{pmid:#1}{\path{#1}}}
\providecommand{\bibinfo}[2]{#2}
\ifx\xfnm\relax \def\xfnm[#1]{\unskip,\space#1}\fi
%Type = Article
\bibitem[{Akizu-Gardoki et~al.(2018)Akizu-Gardoki, Bueno, Wiedmann,
  Lopez-Guede, Arto, Hernandez \& Moran}]{AKIZUGARDOKI20181145}
\bibinfo{author}{Akizu-Gardoki, O.}, \bibinfo{author}{Bueno, G.},
  \bibinfo{author}{Wiedmann, T.}, \bibinfo{author}{Lopez-Guede, J.~M.},
  \bibinfo{author}{Arto, I.}, \bibinfo{author}{Hernandez, P.}, \&
  \bibinfo{author}{Moran, D.} (\bibinfo{year}{2018}).
\newblock \bibinfo{title}{Decoupling between human development and energy
  consumption within footprint accounts}.
\newblock {\it \bibinfo{journal}{Journal of Cleaner Production}\/},  {\it
  \bibinfo{volume}{202}\/}, \bibinfo{pages}{1145 -- 1157}.
%Type = Article
\bibitem[{Chang et~al.(2017)Chang, Lu, Chang \& Lee}]{Chang2017An}
\bibinfo{author}{Chang, G.~W.}, \bibinfo{author}{Lu, H.~J.},
  \bibinfo{author}{Chang, Y.~R.}, \& \bibinfo{author}{Lee, Y.~D.}
  (\bibinfo{year}{2017}).
\newblock \bibinfo{title}{An improved neural network-based approach for
  short-term wind speed and power forecast}.
\newblock {\it \bibinfo{journal}{Renewable Energy}\/},  {\it
  \bibinfo{volume}{105}\/}, \bibinfo{pages}{301--311}.
%Type = Article
\bibitem[{Cui et~al.(2009)Cui, Dang \& Liu}]{Cui2009Novel}
\bibinfo{author}{Cui, J.}, \bibinfo{author}{Dang, Y.~G.}, \&
  \bibinfo{author}{Liu, S.~F.} (\bibinfo{year}{2009}).
\newblock \bibinfo{title}{Novel grey forecasting model and its modeling
  mechanism}.
\newblock {\it \bibinfo{journal}{Control \& Decision}\/},  {\it
  \bibinfo{volume}{24}\/}, \bibinfo{pages}{1702--1706}.
%Type = Article
\bibitem[{Deng(1982)}]{Deng1982Control}
\bibinfo{author}{Deng, J.~L.} (\bibinfo{year}{1982}).
\newblock \bibinfo{title}{Control problems of grey systems}.
\newblock {\it \bibinfo{journal}{Systems \& Control Letters}\/},  {\it
  \bibinfo{volume}{1}\/}, \bibinfo{pages}{288--294}.
%Type = Article
\bibitem[{Ding et~al.(2017)Ding, Dang, Ning, Wei \& Jing}]{Ding2017Multi}
\bibinfo{author}{Ding, S.}, \bibinfo{author}{Dang, Y.~G.},
  \bibinfo{author}{Ning, X.~U.}, \bibinfo{author}{Wei, L.}, \&
  \bibinfo{author}{Jing, Y.~E.} (\bibinfo{year}{2017}).
\newblock \bibinfo{title}{Multi-variable time-delayed discrete grey model}.
\newblock {\it \bibinfo{journal}{Control \& Decision}\/},  {\it
  \bibinfo{volume}{32}\/}, \bibinfo{pages}{199--202}.
%Type = Article
\bibitem[{Ding et~al.(2018)Ding, Dang, Ning \& Zhu}]{Ding2018Modeling}
\bibinfo{author}{Ding, S.}, \bibinfo{author}{Dang, Y.~G.},
  \bibinfo{author}{Ning, X.~U.}, \& \bibinfo{author}{Zhu, X.~Y.}
  (\bibinfo{year}{2018}).
\newblock \bibinfo{title}{Modeling and applications of {DFCGM$(1,N)$} and its
  extended model based on driving factors control}.
\newblock {\it \bibinfo{journal}{Control \& Decision}\/},  {\it
  \bibinfo{volume}{33}\/}, \bibinfo{pages}{712--718}.
%Type = Article
\bibitem[{Du et~al.(2019)Du, Wang \& Niu}]{Du2019hybrid}
\bibinfo{author}{Du, P.}, \bibinfo{author}{Wang, J.~Z.}, \&
  \bibinfo{author}{Niu, T.} (\bibinfo{year}{2019}).
\newblock \bibinfo{title}{A novel hybrid model for short-term wind power
  forecasting}.
\newblock {\it \bibinfo{journal}{Applied Soft Computing}\/},  {\it
  \bibinfo{volume}{80}\/}, \bibinfo{pages}{93--106}.
%Type = Article
\bibitem[{Duan et~al.(2019)Duan, Xiao \& Xiao}]{DUAN2019104853}
\bibinfo{author}{Duan, H.~M.}, \bibinfo{author}{Xiao, X.~P.}, \&
  \bibinfo{author}{Xiao, Q.~Z.} (\bibinfo{year}{2019}).
\newblock \bibinfo{title}{An inertia grey discrete model and its application in
  short-term traffic flow prediction and state determination}.
\newblock {\it \bibinfo{journal}{Neural Computing and Applications}\/}, .
  \DOIprefix\doi{https://doi.org/10.1007/s00521-019-04364-w}.
%Type = Article
\bibitem[{Heuberger \& Dowell(2018)}]{HEUBERGER2018367}
\bibinfo{author}{Heuberger, C.~F.}, \& \bibinfo{author}{Dowell, N.~M.}
  (\bibinfo{year}{2018}).
\newblock \bibinfo{title}{{Real-world} challenges with a rapid transition to
  100\% renewable power systems}.
\newblock {\it \bibinfo{journal}{Joule}\/},  {\it \bibinfo{volume}{2}\/},
  \bibinfo{pages}{367 -- 370}.
%Type = Article
\bibitem[{Hu et~al.(2009)Hu, Zhuang, Zhu \& Fu}]{Rong2009Application}
\bibinfo{author}{Hu, R.}, \bibinfo{author}{Zhuang, Q.~J.},
  \bibinfo{author}{Zhu, L.}, \& \bibinfo{author}{Fu, Y.}
  (\bibinfo{year}{2009}).
\newblock \bibinfo{title}{Application of improved discrete grey model in
  medium-long term power load forecasting}.
\newblock {\it \bibinfo{journal}{Journal of Electric Power Science \&
  Technology}\/},  {\it \bibinfo{volume}{24}\/}, \bibinfo{pages}{49--53}.
%Type = Article
\bibitem[{Jiang et~al.(2018)Jiang, Chen, Guo \& Ding}]{Jiang2018ARIMA}
\bibinfo{author}{Jiang, S.~M.}, \bibinfo{author}{Chen, Y.},
  \bibinfo{author}{Guo, J.~T.}, \& \bibinfo{author}{Ding, Z.~W.}
  (\bibinfo{year}{2018}).
\newblock \bibinfo{title}{{ARIMA} forecasting of {China's} coal consumption,
  price and investment by 2030}.
\newblock {\it \bibinfo{journal}{Energy Sources Part B Economics Planning \&
  Policy}\/},  {\it \bibinfo{volume}{13}\/}, \bibinfo{pages}{1--6}.
%Type = Article
\bibitem[{Jiang et~al.(2012)Jiang, Yan, Feng \& Zhi}]{Wen2012Wind}
\bibinfo{author}{Jiang, W.}, \bibinfo{author}{Yan, Z.}, \bibinfo{author}{Feng,
  D.~H.}, \& \bibinfo{author}{Zhi, H.} (\bibinfo{year}{2012}).
\newblock \bibinfo{title}{Wind speed forecasting using autoregressive moving
  average/generalized autoregressive conditional heteroscedasticity model}.
\newblock {\it \bibinfo{journal}{European Transactions on Electrical Power}\/},
   {\it \bibinfo{volume}{22}\/}, \bibinfo{pages}{662--673}.
%Type = Inproceedings
\bibitem[{Kennedy \& Eberhart(1995)}]{Shi2002Empirical}
\bibinfo{author}{Kennedy, J.}, \& \bibinfo{author}{Eberhart, R.}
  (\bibinfo{year}{1995}).
\newblock \bibinfo{title}{Particle swarm optimization}.
\newblock (pp. \bibinfo{pages}{1942 -- 1948}).
\newblock volume~\bibinfo{volume}{4}.
%Type = Article
\bibitem[{Kiaee et~al.(2018)Kiaee, Infield \& Cruden}]{Kiaee2018Utilisation}
\bibinfo{author}{Kiaee, M.}, \bibinfo{author}{Infield, D.}, \&
  \bibinfo{author}{Cruden, A.} (\bibinfo{year}{2018}).
\newblock \bibinfo{title}{Utilisation of alkaline electrolysers in existing
  distribution networks to increase the amount of integrated wind capacity}.
\newblock {\it \bibinfo{journal}{Journal of Energy Storage}\/},  {\it
  \bibinfo{volume}{16}\/}, \bibinfo{pages}{8--20}.
%Type = Article
\bibitem[{Klessmann et~al.(2011)Klessmann, Held, Rathmann \&
  Ragwitz}]{KLESSMANN20117637}
\bibinfo{author}{Klessmann, C.}, \bibinfo{author}{Held, A.},
  \bibinfo{author}{Rathmann, M.}, \& \bibinfo{author}{Ragwitz, M.}
  (\bibinfo{year}{2011}).
\newblock \bibinfo{title}{Status and perspectives of renewable energy policy
  and deployment in the {European Union-What} is needed to reach the 2020
  targets?}
\newblock {\it \bibinfo{journal}{Energy policy}\/},  {\it
  \bibinfo{volume}{39}\/}, \bibinfo{pages}{7637 -- 7657}.
%Type = Inproceedings
\bibitem[{Liu et~al.(2010)Liu, Li \& Liao}]{Jie2010Non}
\bibinfo{author}{Liu, J.}, \bibinfo{author}{Li, J.~L.}, \&
  \bibinfo{author}{Liao, R.~Q.} (\bibinfo{year}{2010}).
\newblock \bibinfo{title}{Non-equidistance generalized accumulated grey
  forecast model and its application}.
\newblock In {\it \bibinfo{booktitle}{Control Conference}\/}.
%Type = Article
\bibitem[{Liu et~al.(2011)Liu, Peng \& Zhou}]{Liu2011Ship}
\bibinfo{author}{Liu, L.~S.}, \bibinfo{author}{Peng, X.~F.}, \&
  \bibinfo{author}{Zhou, J.~H.} (\bibinfo{year}{2011}).
\newblock \bibinfo{title}{Ship rolling prediction based on gray {RBF} neural
  network}.
\newblock {\it \bibinfo{journal}{Applied Mechanics \& Materials}\/},  {\it
  \bibinfo{volume}{48-49}\/}, \bibinfo{pages}{1044--1048}.
%Type = Article
\bibitem[{Liu et~al.(2018)Liu, Zhang, Chen \& Wang}]{Buhan2015Multi}
\bibinfo{author}{Liu, Y.}, \bibinfo{author}{Zhang, S.}, \bibinfo{author}{Chen,
  X.}, \& \bibinfo{author}{Wang, J.} (\bibinfo{year}{2018}).
\newblock \bibinfo{title}{Artificial combined model based on hybrid nonlinear
  neural network models and statistics linear modelsresearch and application
  for wind speed forecasting}.
\newblock {\it \bibinfo{journal}{Sustainability}\/},  {\it
  \bibinfo{volume}{10}\/}, \bibinfo{pages}{1--30}.
%Type = Article
\bibitem[{Lu et~al.(2019)Lu, Guo, Azimi \& Kun}]{LU201968}
\bibinfo{author}{Lu, H.~F.}, \bibinfo{author}{Guo, L.~J.},
  \bibinfo{author}{Azimi, M.}, \& \bibinfo{author}{Kun, H.}
  (\bibinfo{year}{2019}).
\newblock \bibinfo{title}{Oil and {Gas} 4.0 era: {A} systematic review and
  outlook}.
\newblock {\it \bibinfo{journal}{Computers in Industry}\/},  {\it
  \bibinfo{volume}{111}\/}, \bibinfo{pages}{68 -- 90}.
%Type = Article
\bibitem[{Lundie et~al.(2019)Lundie, Wiedmann, Welzel \&
  Busch}]{LUNDIE20191042}
\bibinfo{author}{Lundie, S.}, \bibinfo{author}{Wiedmann, T.},
  \bibinfo{author}{Welzel, M.}, \& \bibinfo{author}{Busch, T.}
  (\bibinfo{year}{2019}).
\newblock \bibinfo{title}{Global supply chains hotspots of a wind energy
  company}.
\newblock {\it \bibinfo{journal}{Journal of Cleaner Production}\/},  {\it
  \bibinfo{volume}{210}\/}, \bibinfo{pages}{1042 -- 1050}.
%Type = Article
\bibitem[{Ma et~al.(2019{\natexlab{a}})Ma, Cai, Cai \& Dong}]{Ma2019economic}
\bibinfo{author}{Ma, M.~D.}, \bibinfo{author}{Cai, W.}, \bibinfo{author}{Cai,
  W.~G.}, \& \bibinfo{author}{Dong, L.} (\bibinfo{year}{2019}{\natexlab{a}}).
\newblock \bibinfo{title}{Whether carbon intensity in the commercial building
  sector decouples from economic development in the service industry?
  {Empirical} evidence from the top five urban agglomerations in {China}}.
\newblock {\it \bibinfo{journal}{Journal of Cleaner Production}\/},  {\it
  \bibinfo{volume}{222}\/}, \bibinfo{pages}{193--205}.
%Type = Article
\bibitem[{Ma et~al.(2019{\natexlab{b}})Ma, Ma, Cai \& Cai}]{MA2019915}
\bibinfo{author}{Ma, M.~D.}, \bibinfo{author}{Ma, X.}, \bibinfo{author}{Cai,
  W.~G.}, \& \bibinfo{author}{Cai, W.} (\bibinfo{year}{2019}{\natexlab{b}}).
\newblock \bibinfo{title}{Carbon-dioxide mitigation in the residential building
  sector: {A} household scale-based assessment}.
\newblock {\it \bibinfo{journal}{Computers in Industry}\/},  {\it
  \bibinfo{volume}{198}\/}, \bibinfo{pages}{111915}.
  \DOIprefix\doi{https://doi.org/10.1016/j.enconman.2019.111915}.
%Type = Article
\bibitem[{Ma(2019)}]{Ma2019machine}
\bibinfo{author}{Ma, X.} (\bibinfo{year}{2019}).
\newblock \bibinfo{title}{A brief introduction to the {Grey Machine Learning}}.
\newblock {\it \bibinfo{journal}{Journal of Grey System}\/},  {\it
  \bibinfo{volume}{31}\/}, \bibinfo{pages}{1--12}.
%Type = Article
\bibitem[{Ma \& Liu(2017)}]{Ma2017gmc}
\bibinfo{author}{Ma, X.}, \& \bibinfo{author}{Liu, Z.~B.}
  (\bibinfo{year}{2017}).
\newblock \bibinfo{title}{The {GMC(1,n)} model with optimized parameters and
  its application}.
\newblock {\it \bibinfo{journal}{The Journal of Grey System}\/},  {\it
  \bibinfo{volume}{29}\/}, \bibinfo{pages}{122--138}.
%Type = Article
\bibitem[{Ma et~al.(2019{\natexlab{c}})Ma, Wu, Zeng, Wang \& Wu}]{MA2019cfgm}
\bibinfo{author}{Ma, X.}, \bibinfo{author}{Wu, W.~Q.}, \bibinfo{author}{Zeng,
  B.}, \bibinfo{author}{Wang, Y.}, \& \bibinfo{author}{Wu, X.~X.}
  (\bibinfo{year}{2019}{\natexlab{c}}).
\newblock \bibinfo{title}{The conformable fractional grey system model}.
\newblock {\it \bibinfo{journal}{ISA Transactions}\/}, .
  \DOIprefix\doi{https://doi.org/10.1016/j.isatra.2019.07.009}.
%Type = Article
\bibitem[{Ma et~al.(2019{\natexlab{d}})Ma, Xie, Wu, Wu \& Zeng}]{MA2019600}
\bibinfo{author}{Ma, X.}, \bibinfo{author}{Xie, M.}, \bibinfo{author}{Wu,
  W.~Q.}, \bibinfo{author}{Wu, X.~X.}, \& \bibinfo{author}{Zeng, B.}
  (\bibinfo{year}{2019}{\natexlab{d}}).
\newblock \bibinfo{title}{A novel fractional time delayed grey model with {Grey
  Wolf Optimizer} and its applications in forecasting the natural gas and coal
  consumption in {Chongqing China}}.
\newblock {\it \bibinfo{journal}{Energy}\/},  {\it \bibinfo{volume}{178}\/},
  \bibinfo{pages}{487--507}.
%Type = Article
\bibitem[{Ma et~al.(2019{\natexlab{e}})Ma, Xie, Wu, Zeng \&
  Wu}]{Ma2019fractional}
\bibinfo{author}{Ma, X.}, \bibinfo{author}{Xie, M.}, \bibinfo{author}{Wu,
  W.~Q.}, \bibinfo{author}{Zeng, B.}, \& \bibinfo{author}{Wu, X.~X.}
  (\bibinfo{year}{2019}{\natexlab{e}}).
\newblock \bibinfo{title}{The novel fractional discrete multivariate grey
  system model and its applications}.
\newblock {\it \bibinfo{journal}{Applied Mathematical Modelling}\/},  {\it
  \bibinfo{volume}{70}\/}, \bibinfo{pages}{402--424}.
%Type = Article
\bibitem[{Mercer et~al.(2017)Mercer, Sabau \& Klinke}]{Mercer2017}
\bibinfo{author}{Mercer, N.}, \bibinfo{author}{Sabau, G.}, \&
  \bibinfo{author}{Klinke, A.} (\bibinfo{year}{2017}).
\newblock \bibinfo{title}{``wind energy is not an issue for government":
  {Barriers} to wind energy development in {Newfoundland} and {Labrador},
  {Canada}}.
\newblock {\it \bibinfo{journal}{Energy Policy}\/},  {\it
  \bibinfo{volume}{108}\/}, \bibinfo{pages}{673--683}.
%Type = Article
\bibitem[{Moraes et~al.(2018)Moraes, Bussar, Stoecker, Jacqu¨¦, Chang \&
  Sauer}]{Jr2018Comparison}
\bibinfo{author}{Moraes, L.}, \bibinfo{author}{Bussar, C.},
  \bibinfo{author}{Stoecker, P.}, \bibinfo{author}{Jacqu¨¦, K.},
  \bibinfo{author}{Chang, M.}, \& \bibinfo{author}{Sauer, D.~U.}
  (\bibinfo{year}{2018}).
\newblock \bibinfo{title}{Comparison of long-term wind and photovoltaic power
  capacity factor datasets with open-license}.
\newblock {\it \bibinfo{journal}{Applied Energy}\/},  {\it
  \bibinfo{volume}{225}\/}, \bibinfo{pages}{209--220}.
%Type = Book
\bibitem[{Neeru et~al.(2019)Neeru, Srivastava \& Juzer}]{Neeru2019Renewable}
\bibinfo{author}{Neeru, B.}, \bibinfo{author}{Srivastava, V.~K.}, \&
  \bibinfo{author}{Juzer, K.} (\bibinfo{year}{2019}).
\newblock {\it \bibinfo{title}{Renewable energy in {India}: policies to reduce
  greenhouse gas Emissions: challenges, technologies and solutions}\/}.
%Type = Article
\bibitem[{Nong et~al.(2019)Nong, Siriwardana, Perera \& Nguyen}]{NONG201990}
\bibinfo{author}{Nong, D.}, \bibinfo{author}{Siriwardana, M.},
  \bibinfo{author}{Perera, S.}, \& \bibinfo{author}{Nguyen, D.~B.}
  (\bibinfo{year}{2019}).
\newblock \bibinfo{title}{Growth of low emission-intensive energy production
  and energy impacts in {Vietnam} under the new regulation}.
\newblock {\it \bibinfo{journal}{Journal of Cleaner Production}\/},  {\it
  \bibinfo{volume}{225}\/}, \bibinfo{pages}{90 -- 103}.
%Type = Article
\bibitem[{Pali \& Vadhera(2018)}]{Pali2018A}
\bibinfo{author}{Pali, B.~S.}, \& \bibinfo{author}{Vadhera, S.}
  (\bibinfo{year}{2018}).
\newblock \bibinfo{title}{A novel pumped hydro-energy storage scheme with wind
  energy for power generation at constant voltage in rural areas}.
\newblock {\it \bibinfo{journal}{Renewable Energy}\/},  {\it
  \bibinfo{volume}{127}\/}, \bibinfo{pages}{802 -- 810}.
%Type = Article
\bibitem[{Qian et~al.(2012)Qian, Dang \& Liu}]{qian2012}
\bibinfo{author}{Qian, W.~Y.}, \bibinfo{author}{Dang, Y.~G.}, \&
  \bibinfo{author}{Liu, S.~F.} (\bibinfo{year}{2012}).
\newblock \bibinfo{title}{Grey {GM $(1, 1,{t}^{\alpha })$} model with time
  power and its application}.
\newblock {\it \bibinfo{journal}{Engineering Theory \& Practice}\/},  {\it
  \bibinfo{volume}{32}\/}, \bibinfo{pages}{2247--2252}.
%Type = Article
\bibitem[{Safari et~al.(2018)Safari, Chung \& Price}]{Safari2018A}
\bibinfo{author}{Safari, N.}, \bibinfo{author}{Chung, C.~Y.}, \&
  \bibinfo{author}{Price, G.} (\bibinfo{year}{2018}).
\newblock \bibinfo{title}{A novel multi-step short-term wind power prediction
  framework based on chaotic time series analysis and singular spectrum
  analysis}.
\newblock {\it \bibinfo{journal}{IEEE Transactions on Power Systems}\/},  {\it
  \bibinfo{volume}{10}\/}, \bibinfo{pages}{1--9}.
%Type = Article
\bibitem[{Scarlat et~al.(2015)Scarlat, Dallemand, Monforti-Ferrario, Banja \&
  Motola}]{SCARLAT2015969}
\bibinfo{author}{Scarlat, N.}, \bibinfo{author}{Dallemand, J.~F.},
  \bibinfo{author}{Monforti-Ferrario, F.}, \bibinfo{author}{Banja, M.}, \&
  \bibinfo{author}{Motola, V.} (\bibinfo{year}{2015}).
\newblock \bibinfo{title}{Renewable energy policy framework and bioenergy
  contribution in the {European Union} ¨c {An} overview from {National
  Renewable Energy Action Plans} and {Progress Reports}}.
\newblock {\it \bibinfo{journal}{Renewable and Sustainable Energy Reviews}\/},
  {\it \bibinfo{volume}{51}\/}, \bibinfo{pages}{969 -- 985}.
%Type = Article
\bibitem[{Shafiee(2015)}]{Shafiee2015Maintenance}
\bibinfo{author}{Shafiee, M.} (\bibinfo{year}{2015}).
\newblock \bibinfo{title}{Maintenance logistics organization for offshore wind
  energy: {Current} progress and future perspectives}.
\newblock {\it \bibinfo{journal}{Renewable Energy}\/},  {\it
  \bibinfo{volume}{77}\/}, \bibinfo{pages}{182--193}.
%Type = Article
\bibitem[{Sharvini et~al.(2018)Sharvini, Noor, Chong, Stringer \&
  Yusuf}]{SHARVINI2018257}
\bibinfo{author}{Sharvini, S.~R.}, \bibinfo{author}{Noor, Z.~Z.},
  \bibinfo{author}{Chong, C.~S.}, \bibinfo{author}{Stringer, L.~C.}, \&
  \bibinfo{author}{Yusuf, R.~O.} (\bibinfo{year}{2018}).
\newblock \bibinfo{title}{Energy consumption trends and their linkages with
  renewable energy policies in {East} and {Southeast Asian} countries:
  challenges and opportunities}.
\newblock {\it \bibinfo{journal}{Sustainable Environment Research}\/},  {\it
  \bibinfo{volume}{28}\/}, \bibinfo{pages}{257 -- 266}.
%Type = Inproceedings
\bibitem[{Sheng et~al.(2008)Sheng, Zhao \& Lv}]{Sheng2008A}
\bibinfo{author}{Sheng, X.}, \bibinfo{author}{Zhao, H.~F.}, \&
  \bibinfo{author}{Lv, X.~L.} (\bibinfo{year}{2008}).
\newblock \bibinfo{title}{A grey {SVM} based model for patent application
  filings forecasting}.
\newblock In {\it \bibinfo{booktitle}{IEEE International Conference on Fuzzy
  Systems}\/}.
%Type = Article
\bibitem[{Shoaib et~al.(2019)Shoaib, Siddiqui, Rehman, Khan \&
  Alhems}]{SHOAIB2019346}
\bibinfo{author}{Shoaib, M.}, \bibinfo{author}{Siddiqui, I.},
  \bibinfo{author}{Rehman, S.}, \bibinfo{author}{Khan, S.}, \&
  \bibinfo{author}{Alhems, L.} (\bibinfo{year}{2019}).
\newblock \bibinfo{title}{Assessment of wind energy potential using wind energy
  conversion system}.
\newblock {\it \bibinfo{journal}{Journal of Cleaner Production}\/},  {\it
  \bibinfo{volume}{216}\/}, \bibinfo{pages}{346 -- 360}.
%Type = Article
\bibitem[{Shukla et~al.(2017)Shukla, Sudhakar \& Baredar}]{SHUKLA2017342}
\bibinfo{author}{Shukla, A.~K.}, \bibinfo{author}{Sudhakar, K.}, \&
  \bibinfo{author}{Baredar, P.} (\bibinfo{year}{2017}).
\newblock \bibinfo{title}{Renewable energy resources in {South Asian}
  countries: challenges, policy and recommendations}.
\newblock {\it \bibinfo{journal}{Resource-Efficient Technologies}\/},  {\it
  \bibinfo{volume}{3}\/}, \bibinfo{pages}{342 -- 346}.
%Type = Book
\bibitem[{Shumway \& Stoffer(2017)}]{Shumway2017ARIMA}
\bibinfo{author}{Shumway, R.~H.}, \& \bibinfo{author}{Stoffer, D.~S.}
  (\bibinfo{year}{2017}).
\newblock {\it \bibinfo{title}{ARIMA Models. Time Series Analysis and Its
  Applications}\/}.
\newblock \bibinfo{publisher}{Springer New York}.
%Type = Article
\bibitem[{Suddaby(2014)}]{9841836520141001}
\bibinfo{author}{Suddaby, R.} (\bibinfo{year}{2014}).
\newblock \bibinfo{title}{Editor¡¯s comments: Why theory?}
\newblock {\it \bibinfo{journal}{Academy of Management Review}\/},  {\it
  \bibinfo{volume}{39}\/}, \bibinfo{pages}{407 -- 411}.
%Type = Article
\bibitem[{Wang et~al.(2014)Wang, Gong, Zhao \& Zhao}]{Wang2014Extended}
\bibinfo{author}{Wang, H.~T.}, \bibinfo{author}{Gong, L.~H.},
  \bibinfo{author}{Zhao, W.}, \& \bibinfo{author}{Zhao, H.~F.}
  (\bibinfo{year}{2014}).
\newblock \bibinfo{title}{Extended {NGM$(1,1,k)$} model and its application to
  spectral prediction of military launch vehicle hydraulic system}.
\newblock {\it \bibinfo{journal}{Applied Mechanics \& Materials}\/},  {\it
  \bibinfo{volume}{454}\/}, \bibinfo{pages}{90--93}.
%Type = Article
\bibitem[{Wang et~al.(2018{\natexlab{a}})Wang, Du, Lu, Yang \&
  Niu}]{Ding2018production}
\bibinfo{author}{Wang, J.~Z.}, \bibinfo{author}{Du, P.}, \bibinfo{author}{Lu,
  H.~Y.}, \bibinfo{author}{Yang, W.~D.}, \& \bibinfo{author}{Niu, T.}
  (\bibinfo{year}{2018}{\natexlab{a}}).
\newblock \bibinfo{title}{An improved grey model optimized by multi-objective
  ant lionoptimization algorithm for annual electricity consumption
  forecasting}.
\newblock {\it \bibinfo{journal}{Applied Soft Computing Journal}\/},  {\it
  \bibinfo{volume}{72}\/}, \bibinfo{pages}{321--337}.
%Type = Article
\bibitem[{Wang et~al.(2018{\natexlab{b}})Wang, Zhang, Chen \&
  Ma}]{Wang2018nonlinear}
\bibinfo{author}{Wang, Y.}, \bibinfo{author}{Zhang, C.}, \bibinfo{author}{Chen,
  T.}, \& \bibinfo{author}{Ma, X.} (\bibinfo{year}{2018}{\natexlab{b}}).
\newblock \bibinfo{title}{Modeling the nonlinear flow for a multiple-fractured
  horizontal well with multiple finite-conductivity fractures in triple media
  carbonate reservoir}.
\newblock {\it \bibinfo{journal}{Journal of Porous Media}\/},  {\it
  \bibinfo{volume}{21}\/}, \bibinfo{pages}{1283--1305}.
%Type = Article
\bibitem[{Wang \& Li(2019)}]{13285440920190110}
\bibinfo{author}{Wang, Z.~X.}, \& \bibinfo{author}{Li, Q.}
  (\bibinfo{year}{2019}).
\newblock \bibinfo{title}{Modelling the nonlinear relationship between {CO2}
  emissions and economic growth using a {PSO} algorithm-based grey {Verhulst}
  model}.
\newblock {\it \bibinfo{journal}{Journal of Cleaner Production}\/},  {\it
  \bibinfo{volume}{207}\/}, \bibinfo{pages}{214 -- 224}.
%Type = Article
\bibitem[{Wang \& Ye(2017)}]{11978737120170121}
\bibinfo{author}{Wang, Z.~X.}, \& \bibinfo{author}{Ye, D.~J.}
  (\bibinfo{year}{2017}).
\newblock \bibinfo{title}{Forecasting chinese carbon emissions from fossil
  energy consumption using non-linear grey multivariable models}.
\newblock {\it \bibinfo{journal}{Journal of Cleaner Production}\/},  {\it
  \bibinfo{volume}{142}\/}, \bibinfo{pages}{600 -- 612}.
%Type = Article
\bibitem[{Whetten(1989)}]{430837119891001}
\bibinfo{author}{Whetten, D.~A.} (\bibinfo{year}{1989}).
\newblock \bibinfo{title}{What constitutes a theoretical contribution?}
\newblock {\it \bibinfo{journal}{Academy of Management Review}\/},  {\it
  \bibinfo{volume}{14}\/}, \bibinfo{pages}{490 -- 495}.
%Type = Article
\bibitem[{Wu et~al.(2018{\natexlab{a}})Wu, Gao, Xiao, Yang \&
  Chen}]{Wu2018Using}
\bibinfo{author}{Wu, L.~F.}, \bibinfo{author}{Gao, X.~H.},
  \bibinfo{author}{Xiao, Y.~L.}, \bibinfo{author}{Yang, Y.~J.}, \&
  \bibinfo{author}{Chen, X.~G.} (\bibinfo{year}{2018}{\natexlab{a}}).
\newblock \bibinfo{title}{Using a novel multi-variable grey model to forecast
  the electricity consumption of {Shandong} province in {China}}.
\newblock {\it \bibinfo{journal}{Energy}\/},  {\it \bibinfo{volume}{157}\/},
  \bibinfo{pages}{327--335}.
%Type = Article
\bibitem[{Wu et~al.(2013)Wu, Liu, Yao, Yan \& Liu}]{Wu2013Grey}
\bibinfo{author}{Wu, L.~F.}, \bibinfo{author}{Liu, S.~F.},
  \bibinfo{author}{Yao, L.~G.}, \bibinfo{author}{Yan, S.~L.}, \&
  \bibinfo{author}{Liu, D.~L.} (\bibinfo{year}{2013}).
\newblock \bibinfo{title}{Grey system model with the fractional order
  accumulation}.
\newblock {\it \bibinfo{journal}{Communications in Nonlinear Science \&
  Numerical Simulation}\/},  {\it \bibinfo{volume}{18}\/},
  \bibinfo{pages}{1775--1785}.
%Type = Article
\bibitem[{Wu \& Zhang(2018)}]{Wu2018Grey}
\bibinfo{author}{Wu, L.~F.}, \& \bibinfo{author}{Zhang, Z.~Y.}
  (\bibinfo{year}{2018}).
\newblock \bibinfo{title}{Grey multivariable convolution model with new
  information priority accumulation}.
\newblock {\it \bibinfo{journal}{Applied Mathematical Modelling}\/},  {\it
  \bibinfo{volume}{62}\/}, \bibinfo{pages}{595 -- 604}.
%Type = Article
\bibitem[{Wu et~al.(2018{\natexlab{b}})Wu, Ma, Zeng, Wang \&
  Cai}]{Wu2018Application}
\bibinfo{author}{Wu, W.~Q.}, \bibinfo{author}{Ma, X.}, \bibinfo{author}{Zeng,
  B.}, \bibinfo{author}{Wang, Y.}, \& \bibinfo{author}{Cai, W.}
  (\bibinfo{year}{2018}{\natexlab{b}}).
\newblock \bibinfo{title}{Application of the novel fractional grey model
  {FAGMO(1,1,$k$)} to predict {China's} nuclear energy consumption}.
\newblock {\it \bibinfo{journal}{Energy}\/},  {\it \bibinfo{volume}{165}\/},
  \bibinfo{pages}{223--234}.
%Type = Article
\bibitem[{Wu et~al.(2019)Wu, Ma, Zeng, Wang \& Cai}]{Wu2018energy}
\bibinfo{author}{Wu, W.~Q.}, \bibinfo{author}{Ma, X.}, \bibinfo{author}{Zeng,
  B.}, \bibinfo{author}{Wang, Y.}, \& \bibinfo{author}{Cai, W.}
  (\bibinfo{year}{2019}).
\newblock \bibinfo{title}{Forecasting short-term renewable energy consumption
  of {China} using a novel fractional nonlinear grey {Bernoulli} model}.
\newblock {\it \bibinfo{journal}{Renewable Energy}\/},  {\it
  \bibinfo{volume}{140}\/}, \bibinfo{pages}{70--87}.
%Type = Article
\bibitem[{Wu et~al.(2018{\natexlab{c}})Wu, Yong, Lin, Li \&
  Ke}]{Wu2018Identifying}
\bibinfo{author}{Wu, Y.~N.}, \bibinfo{author}{Yong, H.}, \bibinfo{author}{Lin,
  X.}, \bibinfo{author}{Li, L.}, \& \bibinfo{author}{Ke, Y.}
  (\bibinfo{year}{2018}{\natexlab{c}}).
\newblock \bibinfo{title}{Identifying and analyzing barriers to offshore wind
  power development in {China} using the grey decision-making trial and
  evaluation laboratory approach}.
\newblock {\it \bibinfo{journal}{Journal of Cleaner Production}\/},  {\it
  \bibinfo{volume}{189}\/}, \bibinfo{pages}{853 -- 863}.
%Type = Article
\bibitem[{Xia \& Wong(2014)}]{XIA2014119}
\bibinfo{author}{Xia, M.}, \& \bibinfo{author}{Wong, W.}
  (\bibinfo{year}{2014}).
\newblock \bibinfo{title}{A seasonal discrete grey forecasting model for
  fashion retailing}.
\newblock {\it \bibinfo{journal}{Knowledge-Based Systems}\/},  {\it
  \bibinfo{volume}{57}\/}, \bibinfo{pages}{119 -- 126}.
%Type = Article
\bibitem[{Xiao et~al.(2012)Xiao, Luo \& Che}]{Xiao2012Grey}
\bibinfo{author}{Xiao, W.}, \bibinfo{author}{Luo, Y.}, \& \bibinfo{author}{Che,
  X.~Y.} (\bibinfo{year}{2012}).
\newblock \bibinfo{title}{Grey new information unbiased {GRM(1,1)} model based
  on accumulated generating operation in reciprocal number and its
  application}.
\newblock {\it \bibinfo{journal}{Advanced Materials Research}\/},  {\it
  \bibinfo{volume}{426}\/}, \bibinfo{pages}{77--80}.
%Type = Article
\bibitem[{Yang et~al.(2019)Yang, Wang, Lu, Niu \& Du}]{Yang2019hybrid}
\bibinfo{author}{Yang, W.~D.}, \bibinfo{author}{Wang, J.~Z.},
  \bibinfo{author}{Lu, H.~Y.}, \bibinfo{author}{Niu, T.}, \&
  \bibinfo{author}{Du, P.} (\bibinfo{year}{2019}).
\newblock \bibinfo{title}{Hybrid wind energy forecasting and analysis system
  based on divide and conquer scheme: {A} case study in {China}}.
\newblock {\it \bibinfo{journal}{Journal of Cleaner Production}\/},  {\it
  \bibinfo{volume}{222}\/}, \bibinfo{pages}{942--959}.
%Type = Article
\bibitem[{Yin et~al.(2018)Yin, Xu, Li \& Jin}]{YIN2018815}
\bibinfo{author}{Yin, K.~D.}, \bibinfo{author}{Xu, Y.}, \bibinfo{author}{Li,
  X.~M.}, \& \bibinfo{author}{Jin, X.} (\bibinfo{year}{2018}).
\newblock \bibinfo{title}{Sectoral relationship analysis on china's marine-land
  economy based on a novel grey periodic relational model}.
\newblock {\it \bibinfo{journal}{Journal of Cleaner Production}\/},  {\it
  \bibinfo{volume}{197}\/}, \bibinfo{pages}{815 -- 826}.
%Type = Article
\bibitem[{Zendehboudi et~al.(2018)Zendehboudi, Baseer \&
  Saidur}]{ZENDEHBOUDI2018272}
\bibinfo{author}{Zendehboudi, A.}, \bibinfo{author}{Baseer, M.}, \&
  \bibinfo{author}{Saidur, R.} (\bibinfo{year}{2018}).
\newblock \bibinfo{title}{Application of support vector machine models for
  forecasting solar and wind energy resources: {A} review}.
\newblock {\it \bibinfo{journal}{Journal of Cleaner Production}\/},  {\it
  \bibinfo{volume}{199}\/}, \bibinfo{pages}{272 -- 285}.
%Type = Article
\bibitem[{Zeng et~al.(2019)Zeng, Duan \& Zhou}]{ZENG2019385}
\bibinfo{author}{Zeng, B.}, \bibinfo{author}{Duan, H.~M.}, \&
  \bibinfo{author}{Zhou, Y.~F.} (\bibinfo{year}{2019}).
\newblock \bibinfo{title}{A new multivariable grey prediction model with
  structure compatibility}.
\newblock {\it \bibinfo{journal}{Applied Mathematical Modelling}\/},  {\it
  \bibinfo{volume}{75}\/}, \bibinfo{pages}{385 -- 397}.
%Type = Article
\bibitem[{Zeng \& Li(2016)}]{Zeng2018shale}
\bibinfo{author}{Zeng, B.}, \& \bibinfo{author}{Li, C.} (\bibinfo{year}{2016}).
\newblock \bibinfo{title}{Forecasting the natural gas demand in china using a
  self-adapting intelligent grey model}.
\newblock {\it \bibinfo{journal}{Energy}\/},  {\it \bibinfo{volume}{112}\/},
  \bibinfo{pages}{810--825}.
%Type = Article
\bibitem[{Zeng \& Li(2018)}]{Zeng2018background}
\bibinfo{author}{Zeng, B.}, \& \bibinfo{author}{Li, C.} (\bibinfo{year}{2018}).
\newblock \bibinfo{title}{Improved multi-variable grey forecasting model with a
  dynamic background-value coefficient and its application}.
\newblock {\it \bibinfo{journal}{Computers \& Industrial Engineering}\/},  {\it
  \bibinfo{volume}{118}\/}, \bibinfo{pages}{278--290}.
%Type = Article
\bibitem[{Zeng \& Liu(2017)}]{Zeng2017optimal}
\bibinfo{author}{Zeng, B.}, \& \bibinfo{author}{Liu, S.~F.}
  (\bibinfo{year}{2017}).
\newblock \bibinfo{title}{A self-adaptive intelligence gray prediction model
  with the optimal fractional order accumulating operator and its application}.
\newblock {\it \bibinfo{journal}{Mathematical Methods in the Applied
  Sciences}\/},  {\it \bibinfo{volume}{23}\/}, \bibinfo{pages}{1--15}.
%Type = Article
\bibitem[{Zhou(2018)}]{Zhou2018An}
\bibinfo{author}{Zhou, S.~Y.} (\bibinfo{year}{2018}).
\newblock \bibinfo{title}{An exact method for the multiple comparison of
  several polynomial regression models with applications in dose-response
  study}.
\newblock {\it \bibinfo{journal}{Asta Advances in Statistical Analysis}\/},
  {\it \bibinfo{volume}{102}\/}, \bibinfo{pages}{413--429}.
%Type = Article
\bibitem[{Zhou et~al.(2017)Zhou, Zhang, Dang \& Wang}]{Zhou2017New}
\bibinfo{author}{Zhou, W.~J.}, \bibinfo{author}{Zhang, H.~R.},
  \bibinfo{author}{Dang, Y.~G.}, \& \bibinfo{author}{Wang, Z.~X.}
  (\bibinfo{year}{2017}).
\newblock \bibinfo{title}{New information priority accumulated grey discrete
  model and its application}.
\newblock {\it \bibinfo{journal}{Chinese Journal of Management Science}\/},
  {\it \bibinfo{volume}{30}\/}, \bibinfo{pages}{140--148}.

\end{thebibliography}

\end{document}